\documentclass[11pt]{article}

\usepackage[ruled]{algorithm2e}
\SetAlFnt{\small}
\SetAlCapFnt{\small}
\SetAlCapNameFnt{\small}
\SetAlCapHSkip{0pt}
\IncMargin{-\parindent}

\usepackage{algorithmic}
\usepackage{nicefrac}

\usepackage[
  bookmarks=true,
  bookmarksnumbered=true,
  bookmarksopen=true,
  pdfborder={0 0 0},
  colorlinks=true,
  linkcolor=black,
  citecolor=black,
  filecolor=black,
  urlcolor=black,
]{hyperref}

\usepackage{graphicx}
\usepackage{xcolor}
\usepackage[margin=1in]{geometry}
\usepackage{amssymb}
\usepackage{amsmath}
\usepackage[square]{natbib}

\setcitestyle{numbers}
\usepackage{amsthm}
\usepackage[capitalize]{cleveref}
\usepackage{bbm}	

\Crefname{figure}{Figure}{Figures}
\Crefname{equation}{Equation}{Equations}

\newtheorem{theorem}{Theorem}[section]
\newtheorem{corollary}{Corollary}[section]
\newtheorem{lemma}{Lemma}[section]
\newtheorem{proposition}{Proposition}[section]

\newtheorem{observation}{Observation}[section]

\newtheorem{definition}{Definition}[section]
\newtheorem{example}{Example}[section]

\newcommand{\set}[1]{{\left\{#1\right\}}}

\newcommand{\bids}{\bf{b}}
\newcommand{\values}{\bf{v}}

\newcommand{\ud}{unit-demand}
\newcommand{\sm}{submodular}
\newcommand{\xos}{XOS}
\newcommand{\sa}{subadditive}
\newcommand{\mon}{monotone}

\newcommand{\reals}{\mathbb{R}}

\hypersetup{
  pdfauthor      = {Michal Feldman <michal.feldman@cs.tau.ac.il>, Galia Shabtai <galiashabtai@gmail.com>},
  pdftitle       = {Simultaneous 2nd Price Item Auctions with No-Underbidding},
}
\begin{document}

\title{Simultaneous 2nd Price Item Auctions with No-Underbidding
\thanks{This work was partially supported by the Israel Science Foundation (grant number 1219/09), and by the European Research Council (ERC) under the European Union's Horizon 2020 research and innovation program (grant agreement No. 866132)}}

\date{March 15, 2021}

%
%

\author{Michal Feldman\thanks{Tel Aviv University | \emph{E-mail}: \href{mailto:michal.feldman@cs.tau.ac.il}{michal.feldman@cs.tau.ac.il}.} \and Galia Shabtai\thanks{Tel Aviv University | \emph{E-mail}: \href{mailto:galiashabtai@gmail.com}{galiashabtai@gmail.com}.}}


\maketitle

\begin{abstract}
The literature on the Price of Anarchy (PoA) of simple auctions employs a no-overbidding assumption but has completely overlooked the no-underbidding phenomenon, which is evident in empirical studies on variants of the second price auction.  
In this work, we provide a theoretical foundation for the no-underbidding phenomenon. 
We study the PoA of simultaneous 2nd price auctions (S2PA) under a new natural condition of {\em no underbidding}, meaning that agents never bid on items less than their marginal values.
We establish improved (mostly tight) bounds on the PoA of S2PA under no underbidding for different valuation classes (including unit-demand, submodular, XOS, subadditive, and general monotone valuations), in both full-information and incomplete information settings. 


To derive our results, we introduce a new parameterized property of auctions, termed $(\gamma,\delta)$-revenue guaranteed, which implies a PoA of at least $\gamma/(1+\delta)$. 
Via extension theorems, this guarantee extends to coarse correlated equilibria (CCE) in full information settings, and to Bayesian PoA (BPoA) in settings with incomplete information and arbitrary (correlated) distributions.
We then show that S2PA are $(1,1)$-revenue guaranteed with respect to bids satisfying no underbidding. 
This implies a PoA of at least $1/2$ for general monotone valuation, which extends to BPOA with arbitrary correlated distributions.
Moreover, we show that $(\lambda,\mu)$-smoothness combined with $(\gamma,\delta)$-revenue guaranteed guarantees a PoA of at least $(\gamma+\lambda)/(1+\delta+\mu)$.
This implies a host of results, such as a tight PoA of $2/3$ for S2PA with submodular (or XOS) valuations, under no overbidding and no underbidding.
Beyond establishing improved bounds for S2PA, the no underbidding assumption sheds new light on the performance of S2PA relative to 
simultaneous 1st price auctions.
\end{abstract}


\section{Introduction}
\label{sec:intro}
Simple auctions are often preferred in practice over complex truthful auctions.
Starting with the seminal paper of 
\citet{ChristodoulouKS08}, a lot of effort has been given to the study of {\em simultaneous item auctions}. 

In simultaneous item auctions with $n$ bidders and $m$ items, every bidder $i$ has a valuation function $v_i:2^{[m]} \rightarrow \reals^+$, where $v_i(S)$ is the value bidder $i$ assigns to set $S \subseteq [m]$. 
Despite the combinatorial structure of the valuation, bidders submit bids on every item {\em separately and simultaneously}.
In simultaneous first-price auctions (S1PA) every item is sold in a 1st-price auction; i.e., the highest bidder wins and pays her bid, whereas in simultaneous second-price auctions (S2PA) every item is sold in a 2nd-price auction; i.e., the highest bidder wins and pays the 2nd highest bid.

Clearly, these auctions are not truthful; bidders don't even have the language to express their true valuations. 
The performance of these auctions is often quantified by the \emph{price of anarchy} (PoA), which measures their performance in equilibrium. Specifically, the PoA is defined as the ratio between the performance of an auction in its worst equilibrium and the  performance of the optimal outcome. 
The price of anarchy in auctions has been of great interest to the AI community, see the influential survey of 
\citet{RST17}, as well as recent work on learning dynamics in multi-unit auctions and games \cite{FLLST16, BF19}.

\vspace{0.1in}
\noindent
\textbf{PoA and BPoA of S2PA: Background.}
The price of anarchy has been studied both in complete and incomplete information settings. In the former case, all valuations are known by all bidders. In the latter case, every bidder knows her own value and the probability distribution from which other bidder valuations are drawn. The common equilibrium notion in this case is {\em Bayes Nash equilibrium}, and the performance is quantified by the {\em Bayesian PoA} (BPoA) measure.

There are pathological examples showing that the PoA of S2PA can be arbitrarily bad, even in the simplest scenario of a single item auction with two bidders \cite{CKS16}. 
%
A common approach towards overcoming such pathological examples is the {\em no overbidding} (NOB) assumption, stating that the sum of player bids on the set of items she wins 
never exceeds its value.
Consequently, all PoA results of S2PA use the NOB assumption.

The PoA and BPoA of simultaneous item auctions depend on the structure of the valuation functions.
An important class is that of subadditive (\sa)  valuations, also known as {\em complement-free} valuations, where $v(S)+v(T)\geq v(S \cup T)$ for every sets of items $S,T$.
A hierarchy of complement-free valuations is given in \cite{LLN06}, including unit-demand (\ud), submodular (\sm), XOS (\xos), and subadditive (\sa) valuations, with the following strict containment relation: 
$\ud \subset \sm \subset \xos \subset \sa$ (see Section \ref{subsec:val_class} for formal definitions). 
Clearly, the PoA can only degrade as one moves to a larger valuation class.
PoA and BPoA results under the no-overbidding assumption for the different classes have been obtained by \citet{ChristodoulouKS08} and follow-up work \cite{R09,BR11,HassidimKMN11,R12,FFGL13,ST13,CKST16}, and are summarized in Table \ref{tbl:s2pa_nob}. 
\mon\  refers to the class of all monotone valuations.

\begin{table*}[ht]
\centering 
\begin{tabular}{|l|l|lr|lr|lr|lr|} 
\hline 
& & \textbf{UD} $\mathbf{\slash}$ \textbf{SM}& & \textbf{XOS}& & \textbf{SA}& & \textbf{MON}& \\ [0.5ex] 

\hline 
& \textbf{PoA} & $\frac{1}{2}$ & $\leftarrow$& $\frac{1}{2}$ & \cite{CKS16}& $\frac{1}{2}^{\dagger}$& \cite{BR11} & $O \left( \frac{1}{\sqrt{m} } \right)^{\sim}$ & \cite{HassidimKMN11,FFGL13}\\ [0.5ex]
\textbf{NOB} & \textbf{iBPoA} & $\frac{1}{2}$ & $\leftarrow$& $\frac{1}{2}$ & \cite{CKS16}& $\frac{1}{4}^{\sim}$ & \cite{FFGL13}& $O\left(\frac{1}{\sqrt{m}}\right)^{\sim}$& \cite{HassidimKMN11,FFGL13}\\ [0.5ex]
& \textbf{BPoA} & $O \left( \frac{1}{n^{1/4}} \right)^{\sim}$ & \cite{BR11}& $O \left( \frac{1}{n^{1/4}} \right)^{\sim}$&\cite{BR11} &$O \left( \frac{1}{n^{1/4}} \right)^{\sim}$ &\cite{BR11} &$O\left(\frac{1}{n^{1/4}} \right)^{\sim}$ & \cite{BR11}\\ [0.5ex]

\hline 
\end{tabular}
\caption{Previous results for PoA of S2PA under the NOB assumption. PoA is the price of anarchy under full information; iBPoA and BPoA  are the Bayesian PoA under independent valuation distributions and correlated valuation distributions, respectively.
$\dagger$ This result assumes \emph{strong} no-overbidding, i.e. the sum of player bids on \emph{any} set of items never exceeds its value for that set.
All results are tight, except those marked with $\sim$.
Results derived as a special case of a more general result (to their right) are marked with $\leftarrow$.
} 
\label{tbl:s2pa_nob} 
\end{table*}

	
\subsection{No Underbidding (NUB)}

Let $b_i = (b_{i1},\ldots,b_{im})$ denote the bid vector of bidder $i$, where $b_{ij}$ is the bid of bidder $i$ for item $j$, and let $\bids = (b_{1}, \ldots, b_{n})$ be the bid profile of all bidders.
Consider the following example (taken from \citet{CKS16}), showing that the PoA for unit-demand (\ud) valuations is at most $1/2$ (A valuation $v$ is \ud\ if there exist $v(1), \ldots, v(m)$, such that $v(S)=\max_{j\in S}v(j)$ for every set of items $S$).

\begin{example}
	\label{ex:ud-poa-2}
	2 bidders and 2 items, $x,y$. 
	Bidder 1 is \ud\  with values $v_1(x)=2, v_1(y)=1$. Bidder 2 is \ud\  with values $v_2(x)=1, v_2(y)=2$.
	Consider the following bid profile, which is a pure Nash equilibrium (PNE) that adheres to NOB: $b_{1x}=b_{2y}=0$, and $b_{1y}=b_{2x}=1$.
	Under this bid profile, bidders 1 and 2 receive items $y$ and $x$, respectively, for a social welfare of 2. The optimal welfare is 4.
\end{example}

Let us take a closer look at the Nash equilibrium in Example \ref{ex:ud-poa-2}.
In this equilibrium bidder 1 prefers item $x$, yet bids $0$ on item $x$, and gets item $y$ instead.
Bidder 1's marginal value for item $x$, given her current allocation (item $y$), is $v_1(x \mid y)=v_1(xy)-v_1(y)=1$. 
Given her current allocation $y$, bidding $0$ on item $x$ is {\em weakly dominated} by bidding $1$ on $x$. 
Indeed, if bidder 1 receives item $x$, in addition to item $y$, her additional value is $1$ and she pays at most $1$.
Therefore, it is only natural for her to bid at least her marginal value.

If a bidder bids on an item less than the item's marginal value, we say that she {\em underbids} (see Definition \ref{def:underbid_j}).
In Example \ref{ex:ud-poa-2}, bidder 1 underbids on item $x$.
In Section \ref{sec:s2pa_nub} we show that underbidding in a 2nd price auction is {\em weakly dominated} in some precise technical sense.

In what sense is the outcome in Example \ref{ex:ud-poa-2} an equilibrium? While a Nash equilibrium is a descriptive, static notion, it is based on the underlying assumption that players engage in some dynamics, where they keep best responding to the current situation until a stable outcome is reached. 
In this dynamics, it is not likely that a player would bid on an item less than its marginal value. This is exactly what the no-underbidding assumption captures.


No-underbidding is not only a mere theoretical exercise. 
In second price auctions a lot of empirical evidence suggests that bidders tend to overbid, but not underbid  
\cite{KL93,H00,CF08,RS12}. 
It seems that "laboratory second-price auctions exhibit substantial and persistent overbidding, even with prior
experience" \cite{H00}. 
The no-underbidding assumption is also consistent with the assumption made by \citet{nisan2011best} that bidders break ties in favor of the highest bid that does not exceed their value.
Yet, the POA literature employs no-overbidding as a standard assumption, and overlooked the no-underbidding phenomenon. The objective of this work is to better tie the theoretical work in this area to empirical evidence, by providing a theoretical foundation for the no-underbidding phenomenon.

Intuitively, no underbidding can improve welfare performance, as it drives item prices up, so that items become less attractive to low-value players.
Consequently, bad equilibria, in which items are allocated to players with relatively low value, are excluded. 
A natural question is:

\vspace{4mm}

\noindent\textbf{Main Question.} 
{\em What is the performance (measured by PoA/BPoA) of simultaneous 2nd price item auctions under no underbidding?}

\subsection{Our Contribution}
\label{sec:our_contribution}

We first introduce the notion of {\em item no underbidding} (iNUB), where no agent underbids on items (see Definition \ref{def:iNUB}). 
One might think that by imposing both NOB and iNUB, the optimal welfare will be achieved.
This is indeed the case for a single item auction (where the optimal welfare is achieved by imposing any one of these assumptions alone). 
However, even a simple scenario with 2 items and 2 \ud\  bidders can have a PNE with sub-optimal welfare. This is demonstrated in the following example.

\begin{example}
	\label{ex:ud-poa-2-3}
	2 bidders, and 2 items: $x,y$. 
	Bidder 1 is \ud\  with values $v_1(x)=3, v_1(y)=2$. Bidder 2 is \ud\  with values $v_2(x)=2, v_2(y)=3$.
	Consider the following PNE bid profile, which adheres to both NOB and iNUB: $b_{1x}=b_{2y}=1, b_{1y}=b_{2x}=2$.
	Under this bid profile, bidders 1 and 2 receive items $y$ and $x$, respectively, for a social welfare of 4. The optimal welfare is 6. Thus, the PoA is $2/3$. 
\end{example}

\paragraph{{\bf Submodular Valuations}}

Our first result states that $2/3$ is the worst possible ratio for \sm\ valuations and bid profiles satisfying both NOB and iNUB, even in settings with incomplete information (with a product distribution over valuations), see Corollary \ref{cor:sm_smooth_rg_poa}.

\vspace{0.1in}
\noindent {\bf Theorem [\sm\  valuations, iNUB and NOB]:} 
For every market with \sm\  valuations, 
\begin{itemize}
\item
The PoA (even with respect to \emph{coarse correlated equilibrium} (CCE)\footnote{The set of coarse correlated equilibria (CCE) is a superset of Nash equilibria; a formal definition appears in Section \ref{def:ccne}.}) and the BPoA (for product or correlated distribution) of S2PA under iNUB are both at least $\frac{1}{2}$ (see Corollary \ref{cor:sm_rg_poa}).
\item
The PoA (even with respect to CCE), and the BPoA (for product distribution) of S2PA under NOB and iNUB are both at least $\frac{2}{3}$ (see Corollary \ref{cor:sm_smooth_rg_poa}).
\end{itemize}
The above results are tight, even with respect to PNE and even for \ud\  valuations.
\vspace{0.1in}

Moreover, the last theorem extends to $\alpha$-\sm\  valuations, defined as $v(j\mid S) \geq \alpha \cdot v(j\mid T)$ for every $S\subseteq T$. We show that the (B)PoA degrades gracefully with the parameter $\alpha$; namely the PoA with respect to CCE and the BPoA are at least
$\frac{\alpha}{1+\alpha}$ under iNUB and at least 
$\frac{2 \alpha}{2+\alpha}$ under NOB and iNUB (see Corollaries \ref{cor:sm_rg_poa} and \ref{cor:sm_smooth_rg_poa}, respectively). 
\paragraph{{\bf Beyond Submodular Valuations}}

The above bounds do not carry over beyond ($\alpha$-)submodular valuations. 
Consider first XOS valuations, defined as maximum over additive valuations.
We show that the (B)PoA of S2PA with XOS valuations under iNUB is $\theta(\frac{1}{m})$ (the lower bound is given in Appendix \ref{app:s2pa_xos_inub}, and the upper bound is given in Example 
\ref{ex:xos_inub_2_m_poa}). 
Moreover, for XOS valuations, iNUB may not provide any improvement over NOB alone. 
In particular, in Example 
\ref{ex:xos_nob_inub_1_2_poa}
the PoA with NOB and iNUB is $1/2$, matching the guarantee provided by NOB alone. 



To the best of our knowledge, this is the first PoA separation between \sm\  and \xos\  valuations in simultaneous item auctions.
In fact, the PoA of simultaneous item auctions is often the same for the entire range between \ud\  and \xos .  
This separation suggests that \xos\  is ``far" from \sm. Indeed, in Appendix \ref{app:xos_not_alpha_sm} we show that \xos\  is not $\alpha$-\sm\  for any fixed $0< \alpha \leq 1$, even in settings with identical items. 
Moreover, beyond subadditive valuations, the PoA can be arbitrarily bad under bid profiles satisfying iNUB (see Appendix \ref{app:s2pa_mon_inub}).

To deal with valuations beyond \sm, we consider a different no underbidding assumption, which applies to {\em sets of items}.
For two sets $S,T$, the {\em marginal value} of $T$ given $S$ is defined as $v(T \mid S) = v(S \cup T)-v(S)$.  
A bidder is said to not underbid on a set of items $S$ under bid profile $\bids$ if $\sum_{j \in S}b_{ij} \geq v_i(S \mid S_i(\bids))$.
The new condition, {\em set no underbidding} (sNUB), imposes the set no underbidding condition on every bidder $i$ with respect to the set $S = S^*_i(\mathbf{v}) \backslash S_i(\mathbf{b})$ (see Definition \ref{def:sNUB}).

With the sNUB definition, the $2/3$ PoA extends to \sa\  valuations in full information settings, and to \xos\  valuations even in incomplete information settings (with product distributions), see Corollary \ref{cor:xos_smooth_rg_poa_2_3}.

\vspace{0.1in}
\noindent {\bf Theorem [\sa\  and \xos\  valuations, NOB and sNUB]:} For every market with \sa\  valuations, the PoA with respect to CCE of S2PA under strong NOB and sNUB is at least $2/3$ (see Theorem \ref{thm:subadd_nob_snub_cce_2_3}). For every market with \xos\  valuations, the BPoA (under product distribution) of S2PA under NOB and sNUB is at least $2/3$ 
(see Corollary \ref{cor:xos_smooth_rg_poa_2_3}).
Both results are tight.
\vspace{0.1in}
 


For incomplete information we show that the BPoA of \sa\  valuations is at least 
$1/2$ for \emph{any} joint distribution (even correlated) and it can be obtained in a much stronger sense, namely for every bid profile with non-negative sum of utilities (even a non-equilibrium profile) satisfying sNUB. This also holds for markets with arbitrary monotone valuations.

\vspace{0.1in}
\noindent {\bf Theorem [Arbitrary valuations, sNUB]:} For every market (arbitrary monotone valuations), the PoA with respect to CCE and the BPoA (for any joint distribution) of S2PA under sNUB is at least $1/2$ (see Corollary \ref{corr:mon_snub_bpoa_1_2}).
\vspace{0.1in}


The above results are summarized in Table \ref{tbl:s2pa_nub}.


\begin{table}[ht]
\centering 
\begin{tabular}{|l|l|lr|lr|lr|l|} 
\hline 
& & \textbf{UD} $\mathbf{\slash}$ \textbf{SM}& & \textbf{XOS}& & \textbf{SA}& & \textbf{MON} \\ [0.5ex] 

\hline 
\textbf{iNUB} & \textbf{(B)PoA} & $\frac{1}{2}$ &  &  $\Theta \left( \frac{1}{m} \right)$ & &  & & arbitrarily bad \\ [0.5ex]

\hline 
\textbf{sNUB} & \textbf{(B)PoA} & $\frac{1}{2}$  & $\leftarrow$ & $\frac{1}{2}$ & $\leftarrow$ & $\frac{1}{2}$ & $\leftarrow$ & $\frac{1}{2}$ \\ [0.5ex]

\hline 
\textbf{NOB+}& \textbf{PoA} & $\frac{2}{3}$ & $\leftarrow$ & $\frac{2}{3}$ & & $\frac{2}{3}$ & & $\frac{1}{2}$ \\ [0.5ex]
\textbf{sNUB} & \textbf{iBPoA} & $\frac{2}{3}$ & $\leftarrow$ & $\frac{2}{3}$ & & $\frac{1}{2}$ & $\leftarrow$ & $\frac{1}{2}$ \\ [0.5ex]
 & \textbf{BPoA} & $\frac{1}{2}$ & $\leftarrow$ & $\frac{1}{2}$ & $\leftarrow$ & $\frac{1}{2}$ & $\leftarrow$ & $\frac{1}{2}$ \\ [0.5ex]

\hline 
\end{tabular}
\caption{New results for S2PA price of anarchy lower bound. PoA is the price of anarchy under full information, iBPoA is the Bayesian PoA under independent valuation distributions, and BPoA is the Bayesian PoA under correlated valuation distributions.
All results are tight.
Results derived as a special case of a more general result (to their right) are marked with $\leftarrow$.
} 
\label{tbl:s2pa_nub} 
\end{table}


\paragraph{{\bf Equilibrium existence}}
PoA results make sense only when the corresponding equilibrium exists. 
We show that 
every market with \xos\ valuations admits a PNE satisfying sNUB and NOB. For \sa\ valuations, a PNE satisfying NOB might not exist. 
However, under a finite discretized version of the auction, a mixed Bayes Nash equilibrium is guaranteed to exist, and we show that there is at least one bid profile that admits both sNUB and NOB with arbitrary monotone valuation functions.



\paragraph{{\bf S1PA vs. S2PA}}
Interestingly, our results shed new light on the comparison between simultaneous 1st and 2nd price auctions.
Table \ref{tbl:s1pa_vs_s2pa} specifies BPoA lower bounds for S1PA and S2PA under NOB, assuming independent valuation distributions. 
According to these results, one may conclude that S1PA perform better than their S2PA counterparts. 

Our new results shed more light on the relative performance of S2PA and S1PA.
When considering both no overbidding and no underbidding, the situation flips, and S2PA are superior to S1PA.\footnote{Note that no overbidding and no underbidding are not reasonable assumptions in 1st price auctions, where bidders pay their bids}
For \xos\ valuations, the $1-1/e$ bound for S1PA persists, but for S2PA the bound improves from $\frac{1}{2} (<1-\frac{1}{e})$ to $\frac{2}{3} (>1-\frac{1}{e})$. 
For \sa\  valuations and independent valuation distributions, S2PA under sNUB performs as well as S1PA (achieving BPoA of $\frac{1}{2}$), however in S2PA the $\frac{1}{2}$ bound holds also for correlated valuation distributions. For valuations beyond \sa, S2PA performs better ($\frac{1}{2}$ for S2PA and less than $\frac{1}{2}$ for S1PA). 

\begin{table}[ht]
\centering 
\begin{tabular}{|l|l|lr|lr|lr|lr|} 
\hline 
& & \textbf{UD} $\mathbf{\slash}$ \textbf{SM}& & \textbf{XOS}& & \textbf{SA}& & \textbf{MON} &\\ [0.5ex] 

\hline 
\textbf{S2PA} &\textbf{NOB} &  $\frac{1}{2}$ & $\leftarrow$ & $\frac{1}{2}$ & \cite{CKS16}& $\frac{1}{4}$ & \cite{FFGL13} & $O \left( \frac{1}{\sqrt{m} } \right)$ & \cite{HassidimKMN11,FFGL13}\\ [0.5ex]

\hline 
\textbf{S1PA} & &  $1-\frac{1}{e}$ & $\leftarrow$ & $1-\frac{1}{e}$ & \cite{ST13,CKST16}& $\frac{1}{2}$ & \cite{FFGL13,CKST16} & $\frac{1}{m}$ & \cite{HassidimKMN11}\\ [0.5ex]

\hline 
\textbf{S2PA} &\textbf{sNUB+NOB} &  $\frac{2}{3}$ & $\leftarrow$ & $\frac{2}{3}^*$ & & $\frac{1}{2}$ (corr) & $\leftarrow$ & $\frac{1}{2}^*$ (corr)&  \\ [0.5ex]

\hline 
\end{tabular}
\caption{Bayesian price of anarchy results for simultaneous first price and second price auctions. (corr) refers to bounds that hold also for correlated distributions. Results derived from the current paper are marked with *.
Results derived as a special case of a more general result (to their right) are marked with $\leftarrow$.
} 
\label{tbl:s1pa_vs_s2pa} 
\end{table}

\subsection{Our Techniques}
The standard technique for establishing performance guarantees for equilibria of simple auctions (i.e., PoA results) is the {\em smoothness} framework (see the survey in \cite{RST17}).
Smoothness is a parameterized notion; an auction is said to be $(\lambda,\mu)$-smooth if 
for any valuation profile $\mathbf{v}$  and any bid profile $\mathbf{b}$ there exists a bid 
$b^*_i(\mathbf{v})$ 
for each player $i$, s.t. 
$\sum_{i \in [n]} u_i(b^*_i(\mathbf{v}), \mathbf{b}_{-i},v_i) \ge \lambda OPT(\mathbf{v}) - \mu SW(\mathbf{b}, \mathbf{v})$.
{It is quite straightforward to show that if an auction is $(\lambda,\mu)$-smooth, then its PoA with respect to PNE is at least $\frac{\lambda}{1+\mu}$.

The power of the smoothness framework is in its extendability. While a lower bound on the PoA with respect to PNE follows easily from the smoothness property, this lower bound extends
to the PoA with respect to CCE and with respect to the Bayesian PoA in games with incomplete information \cite{R09,R12,RST17,ST13}.


\paragraph{Revenue Guarenteed Auctions}
We introduce a new parameterized notion called {\em revenue guaranteed}. 
An auction is said to be $(\gamma,\delta)$-revenue guaranteed if for every valuation profile $\values$ and bid profile $\bids$ the revenue of the auction is bounded below by $\gamma \cdot OPT(\values) - \delta \cdot SW(\bids,\values)$.

We show that in every $(\gamma,\delta)$-revenue guaranteed auction, the social welfare in every bid profile with non-negative sum of utilities is at least a fraction $\frac{\gamma}{1+\delta}$ of the optimal welfare. Similarly to the smoothness framework, we augment our results with two extension theorems, one for PoA with respect to CCE, and one for BPoA in settings with incomplete information. 
Moreover, this result holds also in cases where the joint distribution of bidder valuations is correlated (whereas previous BPoA results hold only under a product distribution over valuations).

Combining the two tools of smoothness and revenue guaranteed, we get an improved bound. 
In particular, we show that in every auction that is both $(\lambda,\mu)$-smooth and $(\gamma,\delta)$-revenue guaranteed, the PoA with respect to CCE is at least $\frac{\lambda+\gamma}{1+\mu+\delta}$. The same holds for the BPoA under product valuation distributions.

\paragraph{Implications on Simultaneous Second Price Auctions}
With this tool in hand, we analyze simultaneous 2nd price auctions with different valuation functions and different no underbidding conditions, where the goal is to establish revenue-guaranteed parameters that would imply PoA and BPoA bounds. 

We first consider \sm\ and $\alpha$-\sm\ valuations.
We show that every S2PA with $\alpha$-\sm\ valuations satisfying iNUB is $(\alpha,\alpha)$-revenue guaranteed.
This directly gives a lower bound of $\frac{\alpha}{1+\alpha}$ on the BPoA of $\alpha$-\sm\ valuations (and $1/2$ for \sm\ valuations).
We also show that S2PA with $\alpha$-\sm\ valuations satisfying NOB are $(\alpha,1)$-smooth. 
Combining $(\alpha,\alpha)$-revenue guaranteed with $(\alpha,1)$-smoothness gives a bound of $\frac{2 \alpha}{2+\alpha}$ on the BPoA for every S2PA with $\alpha$-\sm\ valuations with NOB and iNUB. For \sm\ valuations this gives the tight $2/3$ bound.

For valuations beyond $\alpha$-\sm\ valuations, the iNUB condition is not helpful, so we turn to the stronger sNUB condition. We show that every S2PA with arbitrary monotone valuations satisfying sNUB is $(1,1)$-revenue guaranteed for bid profiles with non-negative sum of utilities. This recovers the $1/2$ bound on PoA with respect to CCE and BPoA with correlated distributions for S2PA satisfying sNUB. For XOS valuations, we combine the last result with the known $(1,1)$-smoothness to get the $2/3$ bound for S2PA satisfying NOB and sNUB (for PoA with respect to CCE and for BPoA with product distributions). For \sa\ valuations, we apply the technique from \cite{BR11} to yield a tight bound of $2/3$ on the PoA with respect to CCE.


Note that our notion of revenue guarantee is unrelated to the revenue covering property introduced in \citet{HHT14}.
Their revenue covering property is defined for single-parameter settings, and relies heavily on some relationship between thresholds and revenue which does not apply in second price auctions.

\section{Preliminaries}
\label{sec:preliminaries}
\subsection{Auctions}
\label{subsec:auctions}
\paragraph{Combinatorial auctions}
In a combinatorial auction a set of $m$ non-identical items are sold to a group of $n$ players.
Let $\mathcal{S}_i$ be the set of possible allocations to player $i$, $\mathcal{V}_i$ the set of possible valuations of player $i$, and  
$\mathcal{B}_i$ the set of actions available to player $i$. 
Similarly, we let $\mathcal{S} \subseteq \mathcal{S}_1 \times \ldots \times \mathcal{S}_n$ be the allocation space of all players, 
$\mathcal{V} = \mathcal{V}_1 \times \ldots \times \mathcal{V}_n$ be the valuation space, and $\mathcal{B} = \mathcal{B}_1 \times \ldots \times \mathcal{B}_n$ be the action space. An allocation function maps an action profile to an allocation $\mathbf{S} = (S_1, \ldots, S_n) \in \mathcal{S}$, where $S_i$ is the set of items allocated to player $i$. A payment function maps an action profile to a non negative payment $\mathbf{P} = (P_1, \ldots, P_n) \in \mathbb{R}_{+}$, where $P_i$ is the payment of player $i$. We assume that the valuation function $v_i: \mathcal{S}_i \rightarrow \mathbb{R}_{+}$ of a player $i$, where $v_i \in \mathcal{V}_i$, is 
monotone and normalized, i.e., $\forall S \subseteq T \subseteq [m], v_i(S) \le v_i(T)$ and also $v_i(\emptyset)=0$. 
We let $\mathbf{v} = (v_1, \ldots, v_n)$ be the valuation profile. An outcome is 
a pair of allocation $\mathbf{S}$ and payment $\mathbf{P}$
and the revenue is the sum of all payments, i.e. $\mathcal{R}(\mathbf{b}) = \sum_{i \in [n]} P_i(\mathbf{b})$.  
We assume a quasi-linear utility function, i.e. $u_i(S_i, P_i, v_i) = v_i(S_i) - P_i$. We are interested in measuring the \emph{social welfare}, which is the sum of bidder valuations, i.e., $SW(\mathbf{S}, \mathbf{v}) = \sum_{i \in [n]} v_i(S_i)$. 
Given a valuation profile $\mathbf{v}$, an optimal allocation is an allocation that maximizes the $SW$ over all possible allocations. 
We denote by 
$OPT(\mathbf{v})$ the social welfare value of an optimal allocation.

\paragraph{Simultaneous item bidding auction}
In a simultaneous item bidding auction (simultaneous item auction, in short) each item $j \in [m]$ is simultaneously sold in a separate auction. An action profile is a bid profile $\mathbf{b} = (\mathbf{b}_1, \ldots, \mathbf{b}_n)$,
where $b_i = (b_{i1}, \ldots, b_{im})$ is an $m$-vector s.t. $b_{ij}$ is the bid of player $i$ for item $j$. The allocation of each item $j$ is determined by the bids $(b_{1j}, \ldots, b_{nj})$. 
We use $S_i(\mathbf{b})$ to denote the items won by player $i$ and  $p_j(\mathbf{b})$ 
to denote the price paid for item $j$ by the winner of item $j$.
As allocation and payment are uniquely defined by the bid profile, we overload notation and write $u_i(\mathbf{b}, v_i)$ and $SW(\mathbf{b}, \mathbf{v})$.


In a simultaneous second price auction (S2PA), each item $j$ is allocated to the highest bidder, who pays the second highest bid, i.e., $P_i = \sum_{j \in S_i(\mathbf{b})} \max_{k \neq i} b_{kj}$.

In a simultaneous first price auction (S1PA), each item $j$ is allocated to the highest bidder, who pays her bid for that item, i.e., $P_i = \sum_{j \in S_i(\mathbf{b})} b_{ij}$.

Ties are broken arbitrarily but consistently.


\paragraph{Full information setting: solution concepts and PoA}
In the full information setting, the valuation profile $\mathbf{v} = (v_1, \ldots, v_n)$ is known to all players. The standard equilibrium concepts in this setting are pure Nash equilibrium (PNE), mixed Nash equilibrium (MNE), correlated Nash equilibrium (CE) and coarse correlated Nash equilibrium (CCE),
where $PNE \subset MNE \subset CE \subset CCE$. 
Following are the definitions of the equilibrium concepts.
As standard, for a vector $\mathbf{y}$, we denote by $\mathbf{y}_{-i}$ the vector $\mathbf{y}$ with the \emph{i}th component removed.
Also, we denote with $\Delta(\Omega)$ the space of probability distributions over a finite set $\Omega$.

\begin{definition} [\textbf{Pure Nash Equilibrium (PNE)}]
\label{def:pne}
A bid profile 
$\mathbf{b} \in \mathcal{B}_1 \times \ldots \times \mathcal{B}_n$ 
is a PNE if for any $i \in [n]$ and for any $b_i^{'} \in \mathcal{B}_i$,
$u_i(\mathbf{b}, v_i) \ge u_i(b_i^{'},\mathbf{b}_{-i}, v_i)$.
\end{definition}

\begin{definition} [\textbf{Mixed Nash Equilibrium (MNE)}]
\label{def:mne}
A bid profile of randomized bids $\mathbf{b} \in \Delta(\mathcal{B}_1) \times \ldots \times \Delta(\mathcal{B}_n)$ is a MNE if for any $i \in [n]$ and for any $b_i^{'} \in \mathcal{B}_i$,
$\mathbb{E}_{\mathbf{b}} 
\left[ 
u_i(\mathbf{b}, v_i)
\right]
\ge 
\mathbb{E}_{\mathbf{b}_{-i}} 
\left[
u_i(b_i^{'},\mathbf{b}_{-i}, v_i)
\right]$.
\end{definition}

\begin{definition} [\textbf{Correlated Nash Equilibriun (CE)}]
\label{def:cne}
A bid profile of randomized bids $\mathbf{b} \in \Delta(\mathcal{B}_1 \times \ldots \times \mathcal{B}_n)$ is a CE if for any $i \in [n]$ and for any mapping $b_i^{'}(b_i)$,
$\mathbb{E}_{\mathbf{b}} 
\left[ 
u_i(\mathbf{b}, v_i) \mid b_i
\right]
\ge 
\mathbb{E}_{\mathbf{b}} 
\left[
u_i(b_i^{'},\mathbf{b}_{-i}, v_i) \mid b_i
\right]$.
\end{definition}

\begin{definition} [\textbf{Coarse Correlated Nash Equilibriun (CCE)}]
\label{def:ccne}
A bid profile of randomized bids $\mathbf{b} \in \Delta(\mathcal{B}_1 \times \ldots \times \mathcal{B}_n)$ is a CCE if for any $i \in [n]$ and for any $b_i^{'} \in \mathcal{B}_i$,
$\mathbb{E}_{\mathbf{b}} 
\left[ 
u_i(\mathbf{b}, v_i)
\right]
\ge 
\mathbb{E}_{\mathbf{b}} 
\left[
u_i(b_i^{'},\mathbf{b}_{-i}, v_i)
\right]$.
\end{definition}


For a given instance of valuations $\mathbf{v}$, the price of anarchy (PoA) with respect to an equilibrium notion $E$ is defined as: 
$PoA(\mathbf{v}) = \inf_{\mathbf{b} \in E} \frac{\mathbb{E}_{\mathbf{b}} 
\left[SW(\mathbf{b},\mathbf{v}) \right]}{OPT(\mathbf{v})}$. For example, the PoA with respect to PNE is $PoA(\mathbf{v}) = \inf_{\mathbf{b} \in PNE} \frac{SW(\mathbf{b},\mathbf{v})} {OPT(\mathbf{v})}$. The PoA for the other equilibrium types are defined in a similar manner.
For a family of valuations $\mathsf{V}$, 
$PoA(\mathsf{V}) = \min_{\mathbf{v} \in \mathsf{V}} PoA(\mathbf{v})$.



The following lemma will be useful in subsequent sections of this paper.

\begin{lemma}
\label{lem:s2pa_rev_bids}
Consider an S2PA and a valuation $\mathbf{v}$.
Let $S^*(\mathbf{v}) = (S_1^*(\mathbf{v}), \ldots, S_n^*(\mathbf{v}))$
be a welfare-maximizing allocation. Then, for every bid profile $\mathbf{b}$ the following holds:
\begin{eqnarray}
\nonumber
\sum_{i=1}^{n} \sum_{j \in S_i(\mathbf{b})} p_j(\mathbf{b})
&\geq&
\sum_{i=1}^{n} \sum_{j \in S^*_i(\mathbf{v}) \backslash S_i(\mathbf{b})} b_{ij}
\end{eqnarray}
\end{lemma}

\begin{proof}
Let $S^*_{-i}(\mathbf{v}) = \bigcup\limits_{j \ne i} S^*_j(\mathbf{v})$.
Since payments are non-negative, $S_i(\mathbf{b}) \cap S^*_{-i}(\mathbf{v}) \subseteq S_i(\mathbf{b})$, and each item is sold in a separate second price auction, we get:


\begin{eqnarray}
\nonumber
\sum_{i=1}^{n} \sum_{j \in S_i(\mathbf{b})} p_j(\mathbf{b})
&\geq&
\sum_{i=1}^{n} \sum_{j \in S_i(\mathbf{b}) \cap S^*_{-i}(\mathbf{v})} p_j(\mathbf{b}) 
\ \ = \ \  
\sum_{i=1}^{n} \sum_{j \in S_i(\mathbf{b}) \cap S^*_{-i}(\mathbf{v})} max_{k \neq i} ~ b_{kj} \\
\label{eqn:lem:s2pa_rev_bids_1}
&=&
\sum_{i=1}^{n} \sum_{l=1,l \ne i}^{n} \sum_{j \in S_i(\mathbf{b}) \cap S^*_l(\mathbf{v})} max_{k \neq i} ~ b_{kj} 
\ \ \geq \ \ 
\sum_{i=1}^{n} \sum_{l=1,l \ne i}^{n} \sum_{j \in S_i(\mathbf{b}) \cap S^*_l(\mathbf{v})} b_{lj} 
\\
&=&
\nonumber
\sum_{l=1}^{n} \sum_{j \in S^*_l(\mathbf{v}) \backslash S_l(\mathbf{b})} b_{lj}
\end{eqnarray}

Inequality (\ref{eqn:lem:s2pa_rev_bids_1}) holds since $b_{lj}$ is at most the second highest bid on item $j \in S_i(\mathbf{b})$. Notice that the term in (\ref{eqn:lem:s2pa_rev_bids_1}) considers for each player $i$ all the items she wins in bid profile $\mathbf{b}$, which are allocated to some other player $l \ne i$ in the optimal allocation.
Instead, we can change the order of summation and consider for each player $l$ all the items which are allocated to her in the optimal allocation, but not in bid profile $\mathbf{b}$. This accounts for the last equality. 
\end{proof}


\paragraph{Incomplete information setting: solution concepts and Bayesian PoA}
In an incomplete information setting, 
player valuations are drawn from a commonly known, possibly correlated,  joint distribution $\mathcal{F} \in \Delta(\mathcal{V}_1 \times \ldots \times \mathcal{V}_n)$, and the valuation $v_i$ of each player is a private information which is known only to player $i$.
The strategy of player $i$ is a function $\sigma_i: \mathcal{V}_i \rightarrow \mathcal{B}_i$. Let $\Sigma_i$ denote the strategy space of player $i$ and $\Sigma = \Sigma_1 \times \ldots \times \Sigma_n$ the strategy space of all players. 
We denote  by $\mathbf{\sigma}(\mathbf{v}) = (\sigma_1(v_1), \ldots, (\sigma_n(v_n))$ the bid vector given a valuation profile $\mathbf{v}$.

In some cases, we assume that the joint distribution of the valuations is a product distribution, i.e., $\mathcal{F} = \mathcal{F}_1 \times \ldots \times \mathcal{F}_n \in \Delta(\mathcal{V}_1) \times \ldots \times \Delta(\mathcal{V}_n)$. In these cases, 
each valuation $v_i$ is independently drawn from the commonly known distribution $\mathcal{F}_i \in \Delta(\mathcal{V}_i)$.

The standard equilibrium concepts in the incomplete information setting are the Bayes Nash equilibrium (BNE) and the mixed Bayes Nash equilibrium (MBNE): 
\begin{definition} [\textbf{Bayes Nash Equilibriun (BNE)}]
\label{def:bne}
A strategy profile 
$\mathbf{\sigma}$
is a BNE if for any $i \in [n]$, any $v_i \in \mathcal{V}_i$ and any $b_i^{'} \in \mathcal{B}_i$,
\begin{eqnarray}
\label{eqn:bne}
\mathbb{E}_{\mathbf{v}_{-i} 
\mid {v_i}} 
\left[ 
u_i(\sigma_i(v_i), \mathbf{\sigma}_{-i}(\mathbf{v}_{-i}), v_i) 
\right]
\ge 
\mathbb{E}_{\mathbf{v}_{-i} 
\mid {v_i}} 
\left[
u_i(b_i^{'},\mathbf{\sigma}_{-i}(\mathbf{v}_{-i}), v_i)
\right]
\end{eqnarray}
\end{definition}

\begin{definition} [\textbf{Mixed Bayes Nash Equilibriun (MBNE)}]
\label{def:mbne}
A randomized strategy profile 
$\mathbf{\sigma}$ 
is a MBNE if for any $i \in [n]$, any $v_i \in \mathcal{V}_i$ and any $b_i^{'} \in \mathcal{B}_i$,
$$\mathbb{E}_{\mathbf{v}_{-i} 
\mid v_i} 
\mathbb{E}_{\mathbf{\sigma}}
\left[ 
u_i(\sigma_i(v_i), \mathbf{\sigma}_{-i}(\mathbf{v}_{-i}), v_i)
\right]
\ge 
\mathbb{E}_{\mathbf{v}_{-i} \mid v_i} 
\mathbb{E}_{\mathbf{\sigma}_{-i}}
\left[
u_i(b_i^{'},\mathbf{\sigma}_{-i}(\mathbf{v}_{-i}), v_i) 
\right]$$
\end{definition}

Note that if player valuations are independent, we can omit the conditioning on $v_i$ in Definitions \ref{def:bne} and \ref{def:mbne}.

The Bayes Nash price of anarchy is:
$$BPoA = \inf_{\mathcal{F},~\mathbf{\sigma} \in BNE} \frac{\mathbb{E}_{\mathbf{v}} 
\left[SW(\mathbf{\sigma}(\mathbf{v}),\mathbf{v}) \right]}{\mathbb{E}_{\mathbf{v}} 
\left[ OPT(\mathbf{v}) \right]}$$

The mixed Bayes Nash price of anarchy is defined similarly w.r.t. MBNE.

\subsection{Valuation Classes}
\label{subsec:val_class}
In what follows we present the valuation functions considered in this paper. As standard, for a valuation $v$, item $j$ and set $S$, we denote the marginal value of item $j$, given set $S$, as $v(j \mid S)$; i.e., $v(j \mid S) = v(S \cup \set{j}) - v(S)$. In a similar manner, the marginal value of a set $S^{'}$, given a set $S$, is $v(S^{'} \mid S) = v(S \cup S^{'}) - v(S)$. Following are the valuation classes we consider:

\begin{description}
	\item[\ud\ (UD):]
	A valuation function $v$ is UD if there exist values $v_1, \ldots, v_m$ such that for every set $S \subseteq [m]$, $v(S) = max_{j \in S} v_j$.
	\item[\sm\ (SM):]
	A valuation function $v$ is  SM if for
every two sets $S \subseteq T \subseteq [m]$ and element $j \notin T$,
$v(j \mid S) \ge v(j \mid T)$.
	\item[\xos\ (also known as fractionally subadditive):] A valuation function $v$ is XOS if there exists a set $\mathcal{L}$ of additive valuations $\{a_{\ell}(\cdot)\}_{\ell \in \mathcal{L}}$, such that for every set $S \subseteq [m]$, $v(S) = max_{\ell \in \mathcal{L}} a_{\ell}(S)$.
	\item[\sa\ (SA):]
	A valuation function $v$ is SA if for any subsets $S, T \subseteq [m]$, $v(S) + v(T) \ge v(S \cup T)$.
	\item[\mon\ (MON):]
	A valuation function $v$ is MON if $\forall S \subseteq T \subseteq [m], v(S) \le v(T)$.
\end{description}
A strict containment hierarchy of the above valuation classes is known: $UD \subset SM \subset XOS \subset SA \subset MON$.

\subsubsection{\texorpdfstring{$\alpha-$}{}SM Valuation Class}
\label{subsubsec:alpha-sm_val_class}
Classes of set functions are usually characterized by some convenient properties that make them useful in optimization, characterization, approximation, etc..
In practice, the input might only {\em approximately} adhere to some structural property. 
To fall within a particular structural property, fairly stringent constraints should be satisfied. 
The question, then, is whether the guarantees associated with these stringent constraints continue to hold approximately given that these constraints hold approximately. 
This motivates us to introduce a new class of valuation functions, parameterized by `how far' they are from submodular valuations:
\begin{definition}[$\mathbf{\alpha-}$\textbf{\sm\ (}$\mathbf{\alpha-}$\textbf{SM)}]
\label{def:alpha_sm}
A valuation function $v$ is $\alpha-$SM, for $0 < \alpha \le 1$, if for
every two sets $S \subseteq T \subseteq [m]$ and element $j \notin T$,
$v(j \mid S) \ge \alpha \cdot v(j \mid T)$.
\end{definition}



By definition, a SM valuation is precisely $1-$SM. 
Note that the definitions of $\alpha-$SM and XOS are incomparable.
In particular, there is no $\alpha > 0$, such that every XOS function is $\alpha-$SM (see Appendix \ref{app:xos_not_alpha_sm}), and for every $0 < \alpha < 1$, there exists an $\alpha-$SM function that is not XOS (see Appendix \ref{app:alpha_sm_not_xos}).



\begin{lemma}
\label{lem:alpha_sm_prop}
For any $\alpha-$SM function $v$ and any sets $S,S^{'}$: 
$\sum_{j \in S^{'}}v(j \mid S) \geq \alpha \cdot v(S{'} \mid S)$
\end{lemma}

\begin{proof}
Let $S^{'} = \set{j_1, j_2, \ldots, j_{|S^{'}|}}$. 
As $v$ is $\alpha-$SM, we have
$v(j_i \mid S) 
\ge
 \alpha ~ v(j_i \mid S \cup \set{j_1, \ldots, j_{i-1}})$ for every $i=1, \ldots, |S^{'}|$.
Therefore, 
$\sum_{j \in S^{'}} v(j \mid S) 
=
\sum_{i=1}^{|S^{'}|} v(j_i \mid S) 
\ge 
\alpha \sum_{i=1}^{|S^{'}|} v(j_i \mid S \cup \set{j_1, \ldots, j_{i-1}})   
=
\alpha \cdot v(S{'} \mid S)$.
The inequality follows from $\alpha-$submodularity, and the last equality is due to telescoping sum.
\end{proof}

\begin{lemma}
\label{lem:alpha_sm_xos}
If a valuation function, $v$, is $\alpha-$SM, then there 
exists a set $\mathcal{L}$ of additive valuations $\{a^{\ell}(\cdot)\}_{\ell \in \mathcal{L}}$, such that for every set 
$S \subseteq [m]$, $v(S) \ge \alpha \cdot max_{\ell \in \mathcal{L}} \left[ a^{\ell}(S) \right]$ and there exists at least one $\ell$ such that $v(S) = a^{\ell}(S)$.
\end{lemma}

\begin{proof}
The proof is an extension of the proof in \citet{LLN06} that any submodular function is XOS.
Define $m!$ additive valuations $a^{\ell}$, one for each permutation of the items in $[m]$. Let $a^{\ell}_j = v(j \mid S^{\ell}_j)$, where $S^{\ell}_j$ is the set of items in permutation $\ell$ preceding item $j$. 
For any permutation $\ell$ and set $S = \set{1, 2, \ldots, k} \subseteq [m]$ with item $j$ denoting the $j$th item of $S$ in the permutation 
$\ell$,
\begin{eqnarray*}
a^{\ell}(S) 
&=&
\sum_{j \in S} a^{\ell}_j 
\ \ = \ \  
\sum_{j \in S} v(j \mid S^{\ell}_j) \\
&  \le &
\sum_{j \in S}
\frac{1}{\alpha} 
\left[
v(\set{1, 2, \ldots, j}) - v(\set{1, 2, \ldots, j-1})
\right] 
\ \  = \ \ 
\frac{1}{\alpha} v(S),
\end{eqnarray*}
where the inequality follows from the definition of $\alpha-$SM.
For any permutation $\ell$ in which the items of $S$ are placed first, we have $v(S) = a^{\ell}(S)$.
\end{proof}


\subsection{Smooth Auctions}
\label{subsec:smooth}

We use a smoothness definition based on \citet{R09} and \citet{RST17}:

\begin{definition} [\textbf{Smooth auction (based on \cite{R09}, \cite{RST17})}]
\label{def:smooth}
An auction is $(\lambda,\mu)-$smooth for parameters $\lambda,\mu \ge 0$ with respect to a bid space $\mathcal{B}' \subseteq \mathcal{B}_1 \times \ldots \times \mathcal{B}_n$, if for any valuation profile $\mathbf{v} \in \mathcal{V}_1 \times \ldots \times \mathcal{V}_n$ and any bid profile $\mathbf{b} \in \mathcal{B}'$ there exists a bid 
$b^*_i(\mathbf{v}) \in \mathcal{B}_i$ 
for each player $i$, s.t.:
\begin{eqnarray}
\label{eqn:smooth}
\sum_{i \in [n]} u_i(b^*_i(\mathbf{v}), \mathbf{b}_{-i},v_i) \ge \lambda \cdot OPT(\mathbf{v}) - \mu \cdot  SW(\mathbf{b}, \mathbf{v})
\end{eqnarray}
\end{definition}

It is shown in \cite{R09,RST17} that for every $(\lambda,\mu)-$smooth auction, the social welfare of any pure NE is at least $\frac{\lambda}{1+\mu}$. 
Via extension theorems, this bound extends to CCE in full-information settings and to Bayes NE in settings with incomplete information. 
These theorems are stated below, and their proofs appear in Appendix \ref{app:extens_thm_proof} for completeness.


\begin{theorem}
\label{thm:smooth_poa_copm_i}
\textbf{(based on \cite{R09}, \cite{RST17})}
If an auction is $(\lambda,\mu)-$smooth with respect to a bid space $\mathcal{B}' \subseteq \mathcal{B}_1 \times \ldots \times \mathcal{B}_n$, then the expected social welfare of any 
coarse correlated equilibrium, $\mathbf{b} \in \Delta(\mathcal{B}')$, of the auction is at least $\frac{\lambda}{1+\mu}$ of the optimal social welfare.
\end{theorem}



\begin{theorem}
\label{thm:smooth_poa_incopm_i}
\textbf{(based on \cite{R12}, \cite{ST13})}
If an auction is $(\lambda,\mu)-$smooth with respect to a bid space $\mathcal{B}' \subseteq \mathcal{B}_1 \times \ldots \times \mathcal{B}_n$, then 
for 
every product distribution $\mathcal{F}$, 
every mixed Bayes Nash equilibrium, 
$\sigma : \mathcal{V}_1 \times \ldots \times \mathcal{V}_n \rightarrow \Delta(\mathcal{B}')$ has expected social welfare at least $\frac{\lambda}{1+\mu}$ of the expected optimal social welfare.
\end{theorem}

A standard assumption in essentially all previous work on the PoA of simultaneous second price item auction (e.g., \cite{CKS16}, \cite{FFGL13}, \cite{R12}, \cite{R09}, \cite{BR11}) is {\em no overbidding}, meaning that players do not overbid on items they win. 
Formally,

\begin{definition} [\textbf{No overbidding (NOB)}]
\label{def:nob_complete}
Given a valuation profile $\mathbf{v} \in \mathcal{V}_1 \times \ldots \times \mathcal{V}_n$, a bid profile $\mathbf{b} \in \mathcal{B}$ is said to satisfy NOB if for every player $i$ the following holds,
$$\sum_{j \in S_i(\mathbf{b})} b_{ij} \le v_i(S_i(\mathbf{b}))$$
\end{definition}

\begin{theorem} \textbf{(based on \cite{CKS16} and \cite{R09})}:
\label{thm:xos_smooth}
S2PA with XOS valuations is $(1,1)-$smooth, with respect to bid profiles satisfying NOB.   
\end{theorem}




Theorem \ref{thm:xos_smooth} implies a lower bound of $\frac{1}{2}$ on the Bayesian PoA of S2PA with XOS valuations.
This result is tight, even with respect to unit-demand valuations in full information settings \cite{CKS16}.
We now extend the last result to $\alpha$-SM valuations.

\begin{theorem} 
\label{thm:alpha_sm_smooth}
S2PA with $\alpha-$SM valuations is $(\alpha,1)-$smooth, with respect to bid profiles satisfying NOB.   
\end{theorem}

\begin{proof}
Let $\mathbf{v} \in \mathcal{V}_1 \times \ldots \times \mathcal{V}_n$ be 
an $\alpha-$SM valuation profile and let $\mathbf{b}$ be a PNE satisfying NOB. 
From Lemma \ref{lem:alpha_sm_xos}, for every valuation $v_i$
there exists a set $\{a^{\ell}_i(\cdot)\}$, such that for every set 
$S \subseteq [m]$, $v_i(S) \ge \alpha \cdot max_{\ell} \left[ a^{\ell}_i(S) \right]$ and there exists $\ell$ such that $v_i(S) = a^{\ell}_i(S)$.
Let $S^*(\mathbf{v}) = (S_1^*(\mathbf{v}), \ldots, S_n^*(\mathbf{v}))$ be a welfare maximizing allocation, and let $a^*_i$ be an additive valuation such that $v_i(S_i^*(\mathbf{v})) = a^{*}_i(S_i^*(\mathbf{v}))$.  
Consider the following hypothetical deviation for player $i$: $b^*_{ij} = a^*_{ij}$ if $j \in S_i^*(\mathbf{v})$, and $b^*_{ij} = 0$ otherwise.

Now let us consider the utility of player $i$ when deviating. As $b^*_{ij} = 0$ for every item $j \notin S_i^*(\mathbf{v})$, each such item contributes non-negative utility to $i$ and we can ignore this contribution while lower bounding $i$'s utility under $b^*_i$.
Consider item $j \in S_i^*(\mathbf{v})$. 
If $a^*_{ij} \ge \max_{k \ne i} b_{kj}$, player $i$ wins item $j$.
Otherwise, $i$ does not win item $j$, and the term $a^*_{ij} - \max_{k \ne i} b_{kj}$ is non-positive. 
Since $v_i(S) \ge \alpha \cdot a^*_i(S)$ for every set $S$, and since $\alpha \le 1$, we get:
\begin{eqnarray*}
\sum_{i \in [n]}
u_i(b^*_i(\mathbf{v}), \mathbf{b}_{-i},v_i) 
& \ge &
\sum_{i \in [n]}
\left [
\sum_{j \in S_i^*(\mathbf{v}) \cap S_i(\mathbf{b})}
\left(
\alpha \cdot a^*_{ij} - \max_{k \ne i} b_{kj}
\right) + 
\sum_{j \in S_i^*(\mathbf{v}) \backslash S_i(\mathbf{b})}
\left(
a^*_{ij} - \max_{k \ne i} b_{kj}
\right) 
\right] \\
& \ge &
\sum_{i \in [n]}
\sum_{j \in S_i^*(\mathbf{v})}
\left(
\alpha \cdot a^*_{ij} - \max_{k \ne i} b_{kj}
\right) \\
& \ge &
\alpha ~ \sum_{i \in [n]} v_i(S_i^*(\mathbf{v})) - 
\sum_{i \in [n]}
\sum_{j \in S_i^*(\mathbf{v})} 
\max_{k} b_{kj} \\
& \ge &
\alpha \cdot OPT(\mathbf{v}) - 
\sum_{i \in [n]}
\sum_{j \in S_i(\mathbf{b})} 
\max_{k} b_{kj} \\
& = &
\alpha \cdot OPT(\mathbf{v}) - 
\sum_{i \in [n]}
\sum_{j \in S_i(\mathbf{b})} b_{kj} \\
& \ge &
\alpha \cdot OPT(\mathbf{v}) - 
\sum_{i \in [n]} v_i(S_i(\mathbf{b})) \\
& = &
\alpha \cdot OPT(\mathbf{v}) - 
SW(\mathbf{b}, \mathbf{v}).
\end{eqnarray*}

The third inequality follows from the choice of $a^*_i$ and by the payment structure of 2nd price. 
The forth inequality follows by the fact that all items are allocated in equilibrium. Finally, the last inequality follows from NOB. 


\end{proof}

\section{Revenue Guaranteed Auctions}
\label{sec:rev_grnt}

Following is the definition of revenue guaranteed auctions. We then discuss the implications of this property in both full information and incomplete information settings.
\begin{definition} [\textbf{Revenue guaranteed auction}]
\label{def:rev_guarant}
An auction is $(\gamma,\delta)-$revenue guaranteed  for some $0 \le \gamma \le \delta \le 1$ with respect to a bid space $\mathcal{B}' \subseteq \mathcal{B}_1 \times \ldots \times \mathcal{B}_n$, if for any valuation profile $\mathbf{v} \in \mathcal{V}_1 \times \ldots \times \mathcal{V}_n$ and for any bid profile $\mathbf{b} \in \mathcal{B}'$
the revenue of the auction is at least $\gamma \cdot OPT(\mathbf{v}) - \delta \cdot  SW(\mathbf{b}, \mathbf{v})$.
\end{definition}

\subsection{Full Information: Revenue Guaranteed Auctions}
\label{subsec:full_info_rev_grnt}


The following theorem establishes welfare guarantees on every pure bid profile of a $(\gamma,\delta)-$revenue guaranteed auction in which the sum of player utilities is non-negative.

\begin{theorem}
	\label{thm:rev_guarant_poa_copm_i-pure}
	If an auction is $(\gamma,\delta)-$revenue guaranteed with respect to a bid space $\mathcal{B}' \subseteq \mathcal{B}_1 \times \cdots \times \mathcal{B}_n$,
	then for any pure bid profile $\mathbf{b} \in \mathcal{B}'$, 
	in which the sum of player utilities is non-negative, 
	the social welfare is at least $\frac{\gamma}{1+\delta}$ of the optimal social welfare.
\end{theorem}

\begin{proof}
	Using quasi-linear utilities and non-negative sum of player utilities, we get:
	$$0 \le  
	\sum_{i \in [n]}
	u_i(\mathbf{b}, v_i)
	 = 
	\sum_{i \in [n]}
	v_i(S_i(\mathbf{b}))
	 - 
	\sum_{i \in [n]}
	P_i(\mathbf{b})
	 = 
	SW(\mathbf{b}, \mathbf{v})
	 - 
	\sum_{i \in [n]}
	P_i(\mathbf{b})
	$$
	By the $(\gamma,\delta)-$revenue guaranteed property, 
	$$\sum_{i \in [n]} P_i(b) \ge 
	\gamma OPT(\mathbf{v}) - \delta SW(b, \mathbf{v}).$$
	
	Punting it all together, we get
\begin{equation}
	0 
	\ \ \le \ \   
	SW(\mathbf{b}, \mathbf{v})
	- 
	\sum_{i \in [n]}
	P_i(\mathbf{b})
	\ \ \le \ \  
	(1+\delta)SW(\mathbf{b}, \mathbf{v})
	- \gamma OPT(\mathbf{v})
	\label{eqn:rev_guarant_poa_copm_i_2}
\end{equation}
	Rearranging, we get: $SW(\mathbf{b}, \mathbf{v}) \ge \frac{\gamma}{1+\delta} OPT(\mathbf{v})$, as required.
\end{proof}

Definition \ref{def:rev_guarant} considers pure bid profiles, but Theorem~\ref{thm:rev_guarant_poa_copm_i-pure} applies to the more general setting of randomized bid profiles, possibly correlated, as cast in the following extension theorem.

\begin{theorem}
\label{thm:rev_guarant_poa_copm_i}
If an auction is $(\gamma,\delta)-$revenue guaranteed with respect to a bid space $\mathcal{B}' \subseteq \mathcal{B}_1 \times \ldots \times \mathcal{B}_n$,
then for any bid profile $\mathbf{b} \in \Delta(\mathcal{B}')$, 
in which the sum of the expected utilities of the players is non-negative, 
the expected social welfare is at least $\frac{\gamma}{1+\delta}$ of the optimal social welfare.
\end{theorem}

The proof is identical to the proof of Theorem~\ref{thm:rev_guarant_poa_copm_i-pure}, except adding expectation over $\mathbf{b}$ to every term, using the fact the the auction is $(\gamma,\delta)-$revenue guaranteed for every $b$ in the support of $\mathbf{b}$, and using linearity of expectation.

%

Clearly, in every equilibrium (including CCE) the expected utility of every player is non-negative.    
It therefore follows that the expected welfare in any CCE is at least $\frac{\gamma}{1+\delta}$ of the optimal social welfare. 


For an auction that is both smooth and revenue guaranteed, we give a better bound on the price of anarchy:

\begin{theorem}
\label{thm:smooth_rev_guarant_poa_full_info}
If an auction is $(\lambda,\mu)-$smooth with respect to a bid space $\mathcal{B}'$ and $(\gamma,\delta)-$revenue guaranteed with respect to a bid space $\mathcal{B}"$, then the expected social welfare at any 
CCE $\in \Delta(\mathcal{B}' \cap \mathcal{B}")$ of the auction is at least $\frac{\lambda + \gamma}{1+\mu + \delta}$ of the optimal social welfare.
\end{theorem}

\begin{proof}
The proof follows by the proofs of Theorem \ref{thm:smooth_poa_copm_i} and Theorem \ref{thm:rev_guarant_poa_copm_i}.
Let $\mathbf{b} \in \Delta(\mathcal{B}' \cap \mathcal{B}")$ be a CCE of the auction.
The proof of Theorem \ref{thm:smooth_poa_copm_i} shows that:
$$\sum_{i \in [n]}
\mathbb{E}_{\mathbf{b}}
\left[ 
u_i(\mathbf{b}, v_i)
\right] \ge \lambda \cdot OPT(\mathbf{v}) - \mu \cdot \mathbb{E}_{\mathbf{b}}
\left[ 
SW(\mathbf{b}, \mathbf{v})
\right]$$
From Equation (\ref{eqn:rev_guarant_poa_copm_i_2}) we get,
$$\mathbb{E}_{\mathbf{b}}
\left[ 
SW(\mathbf{b}, \mathbf{v})
\right] - 
\sum_{i \in [n]}
\mathbb{E}_{\mathbf{b}}
\left[ 
P_i(\mathbf{b})
\right] \le
(1+\delta) \cdot \mathbb{E}_{\mathbf{b}}
\left[ SW(\mathbf{b}, \mathbf{v})
\right] - \gamma \cdot OPT(\mathbf{v})$$
As utilities are quasi-linear, the left hand side of the above two inequalities are equal. 
Rearranging, we get: $\mathbb{E}_{\mathbf{b}} \left[ SW(\mathbf{b}, \mathbf{v}) \right] \ge \frac{\lambda + \gamma}{1+\mu + \delta} OPT(\mathbf{v})$, as required.
\end{proof}

\subsection{Incomplete Information: Extension Theorem for Revenue Guaranteed Auctions}
\label{subsec:incomp_info_rev_grnt}

In a similar manner to the smoothness extension theorem, we can prove an extension theorem for the revenue guarantee property, which 
gives 
expected 
welfare guarantees for settings with
incomplete information. 
However, this extension theorem is stronger, in the sense that it holds with respect to correlated distributions and not only product prior distributions.

\begin{theorem}
\label{thm:rev_guarant_poa_incopm_i}
If an auction is $(\gamma,\delta)-$revenue guaranteed with respect to a bid space $\mathcal{B}' \subseteq \mathcal{B}_1 \times \cdots \times \mathcal{B}_n$,
then for every 
joint distribution $\mathcal{F} \in \Delta(\mathcal{V}_1 \times \ldots \times \mathcal{V}_n)$, possibly correlated, 
and every strategy profile 
$\sigma : \mathcal{V}_1 \times \ldots \times \mathcal{V}_n \rightarrow \Delta(\mathcal{B}')$,
in which the expected sum of player utilities is non-negative,  
the expected social welfare is at least $\frac{\gamma}{1+\delta}$ of the expected optimal social welfare. 
\end{theorem}

\begin{proof}
We give a proof for pure strategies. The proof for mixed strategies 
follows by adding in a straightforward way 
another expectation over the random actions chosen in the strategy profile $\mathbf{\sigma}$.
As the utility of each player is quasi-linear and the expected sum of player utilities is non-negative, we use linearity of expectation and get,
\begin{eqnarray*}
0 
& \le &  
\sum_{i \in [n]}
\mathbb{E}_{\mathbf{v}}
\left[ 
u_i(\mathbf{\sigma}(\mathbf{v}), v_i)
\right] \\
& = & 
\sum_{i \in [n]}
\mathbb{E}_{\mathbf{v}}
\left[ 
v_i(S_i(\mathbf{\sigma}(\mathbf{v})))
\right] - 
\sum_{i \in [n]}
\mathbb{E}_{\mathbf{v}}
\left[ 
P_i(\mathbf{\sigma}(\mathbf{v}))
\right] \\
& = & 
\mathbb{E}_{\mathbf{v}}
\left[ 
SW(\mathbf{\sigma}(\mathbf{v}), \mathbf{v})
\right] - 
\sum_{i \in [n]}
\mathbb{E}_{\mathbf{v}}
\left[ 
P_i(\mathbf{\sigma}(\mathbf{v}))
\right]
\end{eqnarray*}
By the $(\gamma,\delta)-$revenue guaranteed property, 
for each $v$ in the support of $\mathbf{v}$, 
$$\sum_{i \in [n]} P_i(\mathbf{\sigma}(v)) \ge 
\gamma \cdot OPT(v) - \delta \cdot SW( \mathbf{\sigma}(v),v)$$

Punting it all together, we get,
\begin{eqnarray}
\nonumber
0 
& \le & 
\mathbb{E}_{\mathbf{v}}
\left[ 
SW(\mathbf{\sigma}(\mathbf{v}), \mathbf{v})
\right] - 
\sum_{i \in [n]}
\mathbb{E}_{\mathbf{v}}
\left[ 
P_i(\mathbf{\sigma}(\mathbf{v}))
\right] \\
\label{eqn:rev_guarant_poa_incopm_i_3}
& \le &
\mathbb{E}_{\mathbf{v}}
\left[ 
SW(\mathbf{\sigma}(\mathbf{v}), \mathbf{v})
\right] -
\mathbb{E}_{\mathbf{v}}
\left[ 
\gamma OPT(\mathbf{v}) - \delta SW(\mathbf{\sigma}(\mathbf{v}), \mathbf{v})
\right] \\
& = &
\label{eqn:rev_guarant_poa_incopm_i_2}
(1+\delta) \cdot \mathbb{E}_{\mathbf{v}}
\left[ SW(\mathbf{\sigma}(\mathbf{v}), \mathbf{v})
\right] - \gamma \cdot \mathbb{E}_{\mathbf{v}} \left[ OPT(\mathbf{v}) \right]
\end{eqnarray}

Rearranging, we get: $\mathbb{E}_{\mathbf{v}}
\left[ 
SW(\mathbf{\sigma}(\mathbf{v}), \mathbf{v})
\right] \ge \frac{\gamma}{1+\delta} \mathbb{E}_{\mathbf{v}} \left[ OPT(\mathbf{v}) \right]$, as required.
\end{proof}


As the expected utility of each player is non-negative at any equilibrium strategy profile, we infer that if an auction is $(\gamma,\delta)-$revenue guaranteed with respect to a bid space $\mathcal{B}'$, then for every 
joint distribution $\mathcal{F} \in \Delta(\mathcal{V}_1 \times \ldots \times \mathcal{V}_n)$, possibly correlated, 
the expected social welfare 
at any mixed Bayes Nash equilibrium, $\sigma : \mathcal{V}_1 \times \ldots \times \mathcal{V}_n \rightarrow \Delta(\mathcal{B}')$,
is at least $\frac{\gamma}{1+\delta}$ of the expected optimal social welfare.

For an auction that is both smooth and revenue guaranteed, we give a better bound on the price of anarchy, if the joint distribution $\mathcal{F}$ is a product distribution: 

\begin{theorem}
\label{thm:smooth_rev_guarant_poa_incomp}
If an auction is $(\lambda,\mu)-$smooth with respect to a bid space $\mathcal{B}'$ and $(\gamma,\delta)-$revenue guaranteed with respect to a bid space $\mathcal{B}"$,
then for every product distribution $\mathcal{F}$, 
every mixed Bayes Nash equilibrium, $\sigma : \mathcal{V}_1 \times \ldots \times \mathcal{V}_n \rightarrow \Delta(\mathcal{B}' \cap \mathcal{B}")$, has expected social welfare at least $\frac{\lambda + \gamma}{1+\mu + \delta}$ of the expected optimal social welfare.
\end{theorem}

\begin{proof}
We give a proof for pure Bayes Nash equilibrium. The proof for mixed Bayes Nash equilibrium 
follows by adding in a straightforward way 
another expectation over the random actions chosen in the strategy profile $\mathbf{\sigma}$.
The proof follows by the proofs of Theorem \ref{thm:smooth_poa_incopm_i} and Theorem \ref{thm:rev_guarant_poa_incopm_i}.

The proof of Theorem \ref{thm:smooth_poa_incopm_i} shows that:
$$\mathbb{E}_{\mathbf{v}} 
\left[ 
\sum_{i \in [n]}
u_i(\mathbf{\sigma}(\mathbf{v}), v_i)
\right] 
\ge
\lambda \cdot \mathbb{E}_{\mathbf{v}} 
\left[
OPT(\mathbf{v})
\right] - 
\mu \cdot \mathbb{E}_{\mathbf{v}} 
\left[ SW(\mathbf{\sigma}(\mathbf{v}), \mathbf{v})
\right]$$
From Equation (\ref{eqn:rev_guarant_poa_incopm_i_2}) we get,
$$\mathbb{E}_{\mathbf{v}}
\left[ 
SW(\mathbf{\sigma}(\mathbf{v}), \mathbf{v})
\right] - 
\sum_{i \in [n]}
\mathbb{E}_{\mathbf{v}}
\left[ 
P_i(\mathbf{\sigma}(\mathbf{v}))
\right]
\le 
(1+\delta) \cdot \mathbb{E}_{\mathbf{v}}
\left[ SW(\mathbf{\sigma}(\mathbf{v}), \mathbf{v})
\right] - \gamma \cdot  \mathbb{E}_{\mathbf{v}} \left[ OPT(\mathbf{v}) \right]
$$
As utilities are quasi-linear, the left hand side of the above two inequalities are equal. 
Rearranging, we get: $\mathbb{E}_{\mathbf{v}}
\left[ 
SW(\mathbf{\sigma}(\mathbf{v}), \mathbf{v})
\right] \ge \frac{\lambda + \gamma}{1+\mu + \delta} \mathbb{E}_{\mathbf{v}} \left[ OPT(\mathbf{v}) \right]$, as required.
\end{proof}

\paragraph{\textbf{Remark}}
Note that we defined revenue guaranteed auctions for full information with respect to a bid space $\mathcal{B}'$, and then used an extension theorem to prove positive results for incomplete information on strategy space $\sigma : \mathcal{V}_1 \times \ldots \times \mathcal{V}_n \rightarrow \Delta(\mathcal{B}')$.
A different approach is to add an incomplete information definition to revenue guaranteed auctions with respect to a strategy space $\Sigma'$ and by that get positive results in incomplete information for a wider strategy space.
Moreover, there might be auctions that are revenue guaranteed only in expectation and hence do not fall into the current definition.
We give such definitions and corresponding theorems in Appendix \ref{app:rev_guarant_incomp}.


\section{Simultaneous Second Price Auctions with No-Underbidding}
\label{sec:s2pa_nub}
We begin this section by defining what it means to {\em underbid} on an item. Let $\mathbf{b}_{-j}$ denote the bids of all bidders on items $[m]\setminus \{j\}$.

\begin{definition} [\textbf{item underbidding}]
	\label{def:underbid_j}
	Fix $\mathbf{b}_{-j}$. Player $i$ is said to {\em underbid} on item $j$ if: 
	$b_{ij} < v_i(j \mid S_i(\mathbf{b}_{-j}))$, where $S_i(\mathbf{b}_{-j}) = \set{k \mid k \ne j, b_{ik} = \max_l \set{b_{lk}}}$.
\end{definition}

That is, we say that player $i$ underbids on item $j$ in a bid profile $\bids$ if $i$'s bid on item $j$ is smaller than the marginal valuation of $j$ with respect to the set of items other than $j$ won by $i$.

We next show that underbidding is weakly dominated in a precise sense that we define next.
Consider a bid profile $\mathbf{b}$. Let $\mathbf{b}_{-j}$ be the bids of all bidders on all items except $j$, and let $b_{-ij}$ be the bids on item $j$ of all players, except player $i$. 

\begin{definition} [\textbf{weakly dominated}]
	\label{def:wdominated}
	A bid $b'_{ij}$ is \emph{weakly dominated} by bid $b_{ij}$, with respect to $\mathbf{b}_{-j}$, if the following two conditions hold:
	\begin{enumerate}
		\item
		$u_i(b_{ij},b_{-ij}, \mathbf{b}_{-j};v_i) \ge u_i(b'_{ij},b_{-ij}, \mathbf{b}_{-j};v_i)$, for every $b_{-ij}$
		\label{eq:weakly}
		\item
		There exists $b_{-ij}$ such that the inequality in (\ref{eq:weakly}) holds strictly.
	\end{enumerate}
\end{definition}

The following lemma shows that underbidding on an item in a bid profile is weakly dominated by bidding its marginal value.

\begin{lemma}
\label{lem:i_underbid_wd}
	In S2PA, 
	for every player $i$, every item $j$, and every bid profile $\mathbf{b}_{-j}$, underbidding on item $j$ is weakly dominated by bidding $b_{ij} = v_i(j \mid S_i(b_{-j}))$, with respect to $\mathbf{b}_{-j}$.
\end{lemma}

\begin{proof}
	Fix $\mathbf{b}_{-j}$, and denote $p_j = \max_{l \ne i} \set{b_{lj}}$.  
	Let $b'_{ij}$ be an underbidding bid on item $j$ by player $i$, i.e., $b'_{ij} < v_i(j \mid S_i(\mathbf{b}_{-j}))$.
	We first show that $u_i(b_{ij},b_{-ij}, \mathbf{b}_{-j};v_i) \ge u_i(b'_{ij},b_{-ij}, \mathbf{b}_{-j};v_i)$ for every $b_{-ij}$.
	\begin{itemize}
		\item
		If $b'_{ij} \ge p_j$, player $i$ wins item $j$ under both $b'_{ij}$ and $b_{ij}$, pays $p_j$ on item $j$ in both cases, thus $u_i(b_{ij},b_{-ij}, \mathbf{b}_{-j};v_i) = u_i(b'_{ij},b_{-ij}, \mathbf{b}_{-j};v_i)$.
		
		\item
		If $b_{ij} < p_j$, player $i$ doesn't win item $j$ under both $b'_{ij}$ and $b_{ij}$, pays $0$ on item $j$ in both cases, thus $u_i(b_{ij},b_{-ij}, \mathbf{b}_{-j};v_i) = u_i(b'_{ij},b_{-ij}, \mathbf{b}_{-j};v_i)$. 
		
		\item
		If $b_{ij} \ge p_j$ but $b'_{ij} < p_j$, player $i$ wins item $j$ under $b_{ij}$ and pays $p_j$ on item $j$, but she doesn't win item $j$ under $b'_{ij}$. 
		As $v_i(j \mid S_i(b_{-j})) = b_{ij} \ge p_j$, we have $u_i(b_{ij},b_{-ij}, \mathbf{b}_{-j};v_i) \ge u_i(b'_{ij},b_{-ij}, \mathbf{b}_{-j};v_i)$.
	\end{itemize}
	Next, we show that there exists $b_{-ij}$ such that $u_i(b_{ij},b_{-ij}, \mathbf{b}_{-j};v_i) > u_i(b'_{ij},b_{-ij}, \mathbf{b}_{-j};v_i)$. Let $\epsilon > 0$ be such that $b_{ij} = b'_{ij} + \epsilon$. Consider $b_{-ij}$ and player $l \ne i$ such that $b_{lj} = b'_{ij} + \frac{\epsilon}{2} = p_j$. Then, player $i$ doesn't win item $j$ under $b'_{ij}$, but she does win item $j$ under $b_{ij}$ with a payment of $b'_{ij} + \frac{\epsilon}{2} < b_{ij} = v_i(j|S_i(b_{-j}))$ on item $j$.
	Therefore, $u_i(b_{ij},b_{-ij}, \mathbf{b}_{-j};v_i) > u_i(b'_{ij},b_{-ij}, \mathbf{b}_{-j};v_i)$.
\end{proof}

Motivated by the above analysis, we next define the notion of item no underbidding (iNUB):

\begin{definition} [\textbf{Item No-UnderBidding (iNUB)}]
	\label{def:iNUB}
	Given a valuation profile $\mathbf{v} \in \mathcal{V}_1 \times \ldots \times \mathcal{V}_n$, we say that a bid profile $\mathbf{b} \in \mathcal{B}$ satisfies iNUB if 
	there exists a welfare maximizing allocation, $S^*(\mathbf{v}) = (S_1^*(\mathbf{v}), \ldots, S_n^*(\mathbf{v}))$, such that 
	for every player $i$ 
	and every item 
	$j \in  S^*_i(\mathbf{v})\backslash S_i(\mathbf{b})$ 
	it holds that: 
	$
	b_{ij} \ge v_i(j \mid S_i(\mathbf{b}))
	$.
	\end{definition}


We also define the following notion of set no underbidding (sNUB), as follows: 

\begin{definition} [\textbf{Set No underbidding (sNUB)}]
\label{def:sNUB}
Given a valuation profile $\mathbf{v} \in \mathcal{V}_1 \times \ldots \times \mathcal{V}_n$, we say that a bid profile $\mathbf{b} \in \mathcal{B}$ satisfies sNUB if 
there exists a welfare maximizing allocation, $S^*(\mathbf{v}) = (S_1^*(\mathbf{v}), \ldots, S_n^*(\mathbf{v}))$, such that
for every player $i$, it holds that
\begin{eqnarray}
\nonumber
\sum_{j \in S'} b_{ij} \ge  v_i(S' \mid  S_i(\mathbf{b})), \ \ \ \mbox{where} \ \ \ S' = S^*_i(\mathbf{v})\backslash S_i(\mathbf{b}).
\end{eqnarray}
\end{definition}

In Section \ref{sec:s2pa_sm} we show that if valuations are submodular, every bid profile that satisfies iNUB, also satisfies sNUB.
The opposite is not necessarily true, as demonstrated in Appendix \ref{app:ex_pne_snub_not_inub}.
For unit-demand bidders, iNUB and sNUB coincide (as one can assume w.l.o.g. that every bidder receives a single item in an optimal allocation). 


The following theorem shows that sNUB is a powerful property.

\begin{theorem}
\label{thm:mon_11_rev_guarant}
S2PA with monotone valuation functions is $(1,1)-$revenue guaranteed with respect to 
bid profiles satisfying sNUB.
\end{theorem}

\begin{proof}
In what follows, the first inequality follows by Lemma \ref{lem:s2pa_rev_bids}, 
and the second inequality follows by the fact that $\mathbf{b}$ satisfies sNUB. The last inequality follows by monotonicity of valuations.

\begin{eqnarray*}
\sum_{i=1}^{n} \sum_{j \in S_i(\mathbf{b})} p_j(\mathbf{b})
&\geq&
\sum_{i=1}^{n} \sum_{j \in S^*_i(\mathbf{v}) \backslash S_i(\mathbf{b})} b_{ij} \\
& \ge &
\sum_{i=1}^{n} 
\left[
v_i \left( ~ S_i(\mathbf{b}) \cup \left( S^*_i(\mathbf{v}) \backslash S_i(\mathbf{b})\right)~ \right) - v_i \left( ~ S_i(\mathbf{b})~ \right)
\right] \\
& = &
\sum_{i=1}^{n} 
[v_i(~S_i(\mathbf{b}) \cup S^*_i(\mathbf{v})~) - v_i(~S_i(\mathbf{b})~)] \\
&\geq&
\sum_{i=1}^{n}  [v_i(~S^*_i(\mathbf{v})~) - v_i(~S_i(\mathbf{b})~)] \\
&=&
OPT(\mathbf{v}) - 
SW(\mathbf{b},\mathbf{v})
\end{eqnarray*}
\end{proof}

The following corollary follows directly by Theorems \ref{thm:mon_11_rev_guarant} and \ref{thm:rev_guarant_poa_incopm_i}. 

\begin{corollary}
\label{corr:mon_snub_bpoa_1_2}
In an S2PA with monotone valuations, 
for every 
joint distribution $\mathcal{F} \in \Delta(\mathcal{V}_1 \times \ldots \times \mathcal{V}_n)$, possibly correlated, 
every mixed Bayes Nash equilibrium that satisfies sNUB has expected social welfare at least $\frac{1}{2}$ of the expected optimal social welfare.
\end{corollary}

\paragraph{\textbf{Remark}}
In this section we give a full information definition of no-underbidding bid profiles (sNUB) and prove Theorem \ref{thm:mon_11_rev_guarant} accordingly. In Appendix \ref{app:rev_guarant_incomp} we give an incomplete information definition of no-underbidding strategy profiles, which requires no-underbidding in expectation. This broader definition, together with a broader definition of incomplete information revenue guaranteed auctions, allows us to get positive results in incomplete information setting for a wider strategy space.


\section{S2PA with Submodular Valuations}
\label{sec:s2pa_sm}
In this section we study S2PA with submodular (and $\alpha$-submodular) valuations.
We first show that for this class of valuations, the notion of iNUB suffices for establishing positive results.


\begin{theorem}
\label{thm:sm_rg}
Every S2PA with $\alpha-$SM valuations is $(\alpha,\alpha)-$revenue guaranteed with respect to bid profiles satisfying iNUB.
\end{theorem}

\begin{proof}
	In what follows, the first inequality follows by Lemma \ref{lem:s2pa_rev_bids}, and the second inequality follows by the fact that $\mathbf{b}$ satisfies iNUB.

\begin{eqnarray}
\nonumber
\sum_{i=1}^{n} \sum_{j \in S_i(\mathbf{b})} p_j(\mathbf{b})
&\geq&
\sum_{i=1}^{n} \sum_{j \in S^*_i(\mathbf{v}) \backslash S_i(\mathbf{b})} b_{ij} \\
\nonumber
& \ge & 
\sum_{i=1}^{n} 
\sum_{j \in S^*_i(\mathbf{v}) \backslash S_i(\mathbf{b})} 
v_i(~j \mid S_i(\mathbf{b})~) \\
\label{eqn:sm_rev_guarant_4}
& = &
\sum_{i=1}^{n} 
\sum_{j \in S^*_i(\mathbf{v})} 
v_i(~j \mid S_i(\mathbf{b})~) \\
\label{eqn:sm_rev_guarant_5}
& \ge &
\sum_{i=1}^{n} \alpha \cdot
v_i(~S^*_i(\mathbf{v}) \mid  S_i(\mathbf{b})~) \\
\label{eqn:sm_rev_guarant_6}
&\geq&
\sum_{i=1}^{n} \alpha \cdot [~v_i(~S^*_i(\mathbf{v})~) - v_i(~S_i(\mathbf{b})~)~] \\
\nonumber
&=&
\alpha \cdot OPT(\mathbf{v}) - 
\alpha \cdot SW(\mathbf{b},\mathbf{v})
\end{eqnarray}
Equality (\ref{eqn:sm_rev_guarant_4}) is due to the fact that  $v_i(j \mid S_i(\mathbf{b})) = 0$ for every $j \in S_i(\mathbf{b})$.
Inequality (\ref{eqn:sm_rev_guarant_5}) follows from 
Lemma \ref{lem:alpha_sm_prop}, and Inequality (\ref{eqn:sm_rev_guarant_6}) is due to monotonicity of valuations.
\end{proof}

An immediate corollary from Theorems \ref{thm:rev_guarant_poa_incopm_i} and \ref{thm:sm_rg} is:
\begin{corollary}
\label{cor:sm_rg_poa}
In an S2PA with $\alpha-$SM valuations, 
for every 
joint distribution $\mathcal{F} \in \Delta(\mathcal{V}_1 \times \ldots \times \mathcal{V}_n)$, possibly correlated, 
and every strategy profile $\mathbf{\sigma}$ 
that satisfies iNUB for which the expected sum of player utilities is non-negative, 
the expected social welfare is at least $\frac{\alpha}{1+\alpha}$ of the expected optimal social welfare.
In particular, for SM valuations (where $\alpha=1$), we get at least $\frac{1}{2}$ of the expected optimal social welfare.
\end{corollary}


The $1/2$ bound for submodular valuations is tight, even with respect to unit-demand valuations and even in equilibrium, as  shown in the following proposition.

\begin{proposition}
\label{prop:sm_wnub_poa_eq_1_2}
There exists an S2PA with unit-demand valuations that admits a PNE bid profile that satisfies iNUB,
where the social welfare in equilibrium is $\frac{1}{2}$ of the optimal social welfare.
\end{proposition}

\begin{proof}
Consider an S2PA with two unit demand players 
and 2 items, $\set{x,y}$, where $v_1(x) = 2$, $v_1(y) = 1$, $v_2(x) = 1$ and $v_2(y) = 2$. 
An optimal allocation gives item $x$ to player $1$ and item $y$ to player $2$, for a welfare of $4$.
Consider the following bid profile $\mathbf{b}$: 
$b_{1x} = 1$, $b_{1y} = 100$, $b_{2x} = 100$ and $b_{2y} = 1$. Player $1$ wins item $y$ for a price of $1$,  and player $2$ wins item $x$ for a price of $1$. It is easy to see that $\mathbf{b}$ is a PNE that satisfies iNUB. 
The social welfare of this equilibrium is $2$, which is $\frac{1}{2}$ of the optimal social welfare.
\end{proof}

We next show that for submodular valuations, iNUB implies sNUB. 

\begin{proposition}
\label{prop:sm_inub_snub}
For every SM valuation $\mathbf{v}$, every bid profile $\mathbf{b}$ that satisfies iNUB also satisfies sNUB.
\end{proposition}


\begin{proof}
	By iNUB, for every item $j \in S_i^*(\mathbf{v})\setminus S_i(\mathbf{b})$, it holds that $b_{ij} \geq v_i(j \mid S_i(\mathbf{b}))$. 
	It follows that
	$$
	\sum_{j \in S_i^*(\mathbf{v})\setminus S_i(\mathbf{b})}b_{ij} \geq 
	\sum_{j \in S_i^*(\mathbf{v})\setminus S_i(\mathbf{b})}v_i(j \mid S_i(\mathbf{b})) 
	\ge v_i(S_i^*(\mathbf{v})\setminus S_i(\mathbf{b}) \mid S_i(\mathbf{b})),
	$$ 
	where the last inequality follows by Lemma \ref{lem:alpha_sm_prop} for $\alpha=1$. 
\end{proof}


Therefore, Theorem \ref{thm:mon_11_rev_guarant} and Corollary \ref{corr:mon_snub_bpoa_1_2} also apply to every SM valuation that satisfies iNUB. That is,

\begin{corollary}
\label{{cor:sm_inub_11_rev_guarant}}
		S2PA with submodular valuation functions is $(1,1)-$revenue guaranteed with respect to 
		bid profiles satisfying iNUB.
\end{corollary}

\begin{corollary}
\label{corr:sm_inub_bpoa_1_2}
In an S2PA with submodular valuations, 
for every 
joint distribution $\mathcal{F} \in \Delta(\mathcal{V}_1 \times \ldots \times \mathcal{V}_n)$, possibly correlated, 
every mixed Bayes Nash equilibrium that satisfies iNUB has expected social welfare at least $\frac{1}{2}$ of the expected optimal social welfare.
\end{corollary}

For bid profiles that satisfy both iNUB and NOB, we get a better bound as a direct corollary from Theorems
\ref{thm:sm_rg},  
\ref{thm:alpha_sm_smooth}, and
\ref{thm:smooth_rev_guarant_poa_incomp}.

\begin{corollary}
\label{cor:sm_smooth_rg_poa}
In an S2PA with $\alpha-$SM valuations,
for 
every product distribution $\mathcal{F}$, 
every mixed Bayes Nash equilibrium that satisfies both NOB and iNUB
has expected social welfare at least $\frac{2 \alpha}{2+\alpha}$ of the expected optimal social welfare. In particular, for SM valuations (where $\alpha=1$) this amounts to at least $\frac{2}{3}$ of the expected optimal social welfare, and this is tight. 
\end{corollary}

The bound of $2/3$ for submodular valuations is tight even with respect to a PNE with unit-demand valuations. This is shown in Example \ref{ex:ud-poa-2-3} in Section \ref{sec:intro}.


A remark about existence of pure NE is in order. In the next section, we show that every S2PA with XOS valuations admits a pure NE that satisfies both NOB and sNUB (Theorem \ref{thm:pne_xos}). Since every submodular valuation is XOS, the existence result applies also to submodular valuations. 

\section{S2PA with XOS Valuations}
\label{sec:s2pa_xos}
\subsection{XOS Valuations under iNUB}
\label{sec:s2pa_xos_inub}

For XOS valuations, iNUB does not imply sNUB, thus iNUB does not lead automatically to PoA lower bounds. 
Indeed, Example~\ref{ex:xos_inub_2_m_poa} shows an instance of S2PA with XOS valuations that admits a PNE bid profile that satisfies iNUB, with social welfare that is only a $\frac{2}{m}$ fraction of the optimal social welfare. 


\begin{example}
\label{ex:xos_inub_2_m_poa}
Consider an S2PA with two players, with the following XOS valuation functions $v_1$ and $v_2$, respectively, over $m$ items:
\begin{eqnarray*}
v_1(S) &=& max \set{a^1(S), a^2(S)}, \\
v_2(S) &=& max \set{a^3(S), a^4(S)}
\end{eqnarray*}

where:
\begin{eqnarray*}
a^1 
&=& 
(a^1_1, a^1_2, \ldots, a^1_m) 
=
(2,2,0,0,0,0, \ldots, 0), \\
a^2 
&=& 
(a^2_1, a^2_2, \ldots, a^2_m) 
=
(0,0,1,1,0,0, \ldots, 0), \\
a^3 
&=& 
(a^3_1, a^3_2, \ldots, a^3_m) 
=
(0,0,2,2,2,2, \ldots, 2), \\
a^4 
&=& 
(a^4_1, a^4_2, \ldots, a^4_m) 
=
(1,1,0,0,0,0, \ldots, 0).
\end{eqnarray*}
The optimal allocation has welfare $2m$, giving the first two items to player $1$ and the last $m-2$ items to player $2$.
Consider the following bid profile $\mathbf{b} = (b_1,b_2)$, where:
\begin{eqnarray*}
b_1 &=& (0,0,2,2,2,2, \ldots, 2) \\
b_2 &=& (2,2,0,0,0,0, \ldots, 0)
\end{eqnarray*}
Player $2$ wins the first two items and player $1$ wins the last $m-2$ items.
It is easy to see that $\mathbf{b}$ is an equilibrium which satisfies iNUB. 
The social welfare of this equilibrium is $4$, which is $\frac{2}{m}$ of the optimal social welfare.
\end{example}

\subsection{XOS Valuations under sNUB}
\label{sec:s2pa_xos_snub}
As Theorem \ref{thm:mon_11_rev_guarant} and Corollary \ref{corr:mon_snub_bpoa_1_2} apply for arbitrary monotone valuation functions, the PoA is at least $\frac{1}{2}$ with respect to bid profiles satisfying sNUB. 
An immediate corollary from Theorems 
\ref{thm:xos_smooth},
\ref{thm:mon_11_rev_guarant} and
\ref{thm:smooth_rev_guarant_poa_incomp}  is:

\begin{corollary}
\label{cor:xos_smooth_rg_poa_2_3}
In an S2PA with XOS valuations,
for 
every product distribution $\mathcal{F}$, 
every mixed Bayes Nash equilibrium that satisfies both NOB and sNUB has expected social welfare at least $\frac{2}{3}$ of the expected optimal social welfare.
\end{corollary}

As in the case of submodular valuations, this result is tight (see Example \ref{ex:ud-poa-2-3}).



As for existence of equilibria, we next show that every S2PA with XOS valuations admits a pure NE that satisfies both NOB and sNUB. 
 
\begin{theorem}
\label{thm:pne_xos}
In S2PA with XOS valuations there always exists at least one pure Nash equilibrium that satisfies both NOB and sNUB.
\end{theorem}


\begin{proof}
\citet{CKS16} showed that every S2PA with XOS valuations admits a PNE satisfying NOB. We show that the same PNE satisfies sNUB as well.
Let $S^*(\mathbf{v}) = (S_1^*(\mathbf{v}), \ldots, S_n^*(\mathbf{v}))$ be a welfare maximizing allocation, and let $a_i^*$ be an additive valuation such that $v_i(S_i^*(\mathbf{v})) = \sum_{j \in S_i^*(\mathbf{v})} a^*_{ij}$. Consider the bid profile in which every player bids according to the maximizing additive valuation with
respect to her set $S_i^*(\mathbf{v})$, i.e., $b_{ij} = a^*_{ij}$ for every $j \in S_i^*(\mathbf{v})$ and $b_{ij} = 0$ otherwise. One can easily verify that this bid profile is a PNE that satisfies NOB. 
It thus remains to show that it also satisfies sNUB.
Recall that sNUB imposes restrictions on the bid values of the set $S' = S_i^*(\mathbf{v}) \backslash S_i(\mathbf{b})$.
Under the above bid profile we have $S_i(\mathbf{b}) = S_i^*(\mathbf{v})$, i.e., $S' = \emptyset$ and sNUB holds trivially.
\end{proof}

\section{S2PA with Subadditive Valuations}
\label{sec:s2pa_sa}
Recall that S2PA with arbitrary monotone valuations is $(1,1)-$revenue guaranteed with respect to bid profiles that satisfy sNUB (Theorem \ref{thm:mon_11_rev_guarant}). 
Hence, the Bayesian PoA for equilibria satisfying sNUB is at least $\frac{1}{2}$ and this bound is tight (Proposition \ref{prop:sm_wnub_poa_eq_1_2}). 


For subadditive valuations, \citet{BR11} showed that the social welfare of any PNE (if it exists) satisfying \emph{strong} NOB is at least $\frac{1}{2}$ of the optimal social welfare.\footnote{
In the \emph{strong} no-overbidding assumption
the sum of bids on \emph{any} set of items does not exceed the value of that set, i.e., for every $i$ and every subset $S \subseteq [m]$, $\sum_{j \in S} b_{ij} \le v_i(S)$. 
} They also showed that this bound is tight.
In their proof, they used the fact that the revenue is non-negative. Under our sNUB condition, the auction is $(1,1)-$revenue guaranteed, implying that the revenue is lower bounded by $OPT(\mathbf{v}) - SW(\mathbf{b},\mathbf{v})$. Plugging this lower bound into their proof, we get:

\begin{theorem}
	\label{thm:subadd_nob_snub_2_3}
	In an S2PA with subadditive valuations and at least one
	pure Nash equilibrium that satisfies both strong NOB and sNUB, the social welfare of such a PNE is at least $\frac{2}{3}$ of the expected optimal social welfare.
\end{theorem}

\begin{proof}
	\citet{BR11} showed that for subadditive valuations the following holds for any PNE, $\mathbf{b}$, satisfying strong NOB, if exists\footnote{This may look as a smoothness argument. However, while the hypothetical deviation considered in the smoothness proof depends on $\mathbf{v}$ but not on $\mathbf{b}$, the proof in \cite{BR11} invokes the Nash equilibrium hypothesis for player $i$ with an hypothetical deviation that depends on the bid vectors $\mathbf{b}_{-i}$ of the other players.}
	$$\sum_{i \in [n]} u_i(\mathbf{b},v_i) \ge OPT(\mathbf{v}) - SW(\mathbf{b}, \mathbf{v})$$
	Since we assume that $\mathbf{b}$ also satisfies sNUB, we have that:
	$$\sum_{i=1}^{n} \sum_{j \in S_i(\mathbf{b})} p_j(\mathbf{b})
	\geq
	OPT(\mathbf{v}) - 
	SW(\mathbf{b},\mathbf{v})
	$$
	As utilities are quasi-linear, we get,
	\begin{eqnarray*}
		\sum_{i \in [n]} u_i(\mathbf{b},v_i) & = &
		SW(\mathbf{b},\mathbf{v}) - \sum_{i=1}^{n} \sum_{j \in S_i(\mathbf{b})} p_j(\mathbf{b}) \\
		& \le &
		SW(\mathbf{b},\mathbf{v}) - \left[
		OPT(\mathbf{v}) - 
		SW(\mathbf{b},\mathbf{v}) 
		\right] \\
		& = & 
		2 \cdot SW(\mathbf{b},\mathbf{v}) - OPT(\mathbf{v})
	\end{eqnarray*}
	Putting it all together, we get:
	\begin{eqnarray*}
		2 \cdot SW(\mathbf{b},\mathbf{v}) - OPT(\mathbf{v})
		& \ge &
		\sum_{i \in [n]} u_i(\mathbf{b},v_i) \\
		& \ge &
		OPT(\mathbf{v}) - SW(\mathbf{b}, \mathbf{v})
	\end{eqnarray*}
	Rearranging, we get: $SW(\mathbf{b}, \mathbf{v}) \ge \frac{2}{3} OPT(\mathbf{v})$, as required.
\end{proof}

\citet{BR11} generalized the results from PNE to CCE and showed that in an S2PA with subadditive valuations  the social welfare of every CCE satisfying strong NOB is at least $\frac{1}{2}$ of the optimal social welfare. Adding the assumption that the CCE satisfies also sNUB, it is straightforward to show (in a similar manner to the proof of Theorem \ref{thm:subadd_nob_snub_2_3} ) that:

\begin{theorem}
\label{thm:subadd_nob_snub_cce_2_3}
	In an S2PA with subadditive valuations, the social welfare of every CCE satisfying both strong NOB and sNUB, is at least $\frac{2}{3}$ of the optimal social welfare.
\end{theorem}

Example \ref{ex:ud-poa-2-3} shows that the above bound is tight. 
%

\citet{FFGL13} proved that in an S2PA with independent subadditive valuations the expected social welfare of every mixed Bayes Nash equilibrium which satisfies NOB, is at least $\frac{1}{4}$ of the optimal social welfare. Adding the sNUB assumption to the bid profile and using the fact the auction is $(1,1)-$revenue guaranteed (as in the proof of Theorem \ref{thm:subadd_nob_snub_2_3}) improves the lower bound on the BPoA to $\frac{1}{2}$ for independent subadditive valuations with both NOB and sNUB. Recall that the same lower bound is obtained without the NOB or equilibrium assumptions (see Corollary \ref{corr:mon_snub_bpoa_1_2}). We conclude that NOB does not improve the BPoA bound in this case. 

\citet{BR11} gave an example of S2PA with subadditive valuations that does not have a PNE satisfying strong NOB.
\citet{FKL12} gave conditions on a valuation profile under which no conditional equilibrium exists, i.e. no S2PA PNE with NOB exists.
\citet{DHS18} proved that it is NP-hard to decide whether there exists a PNE in VCG mechanisms for agents with subadditive valuations and
additive bids. That is, a PNE in S2PA with subadditive valuations is not guaranteed to exist. 

In general, a mixed equilibrium may not exist in infinite games (as in our case of continuous valuations and continuous bids). However, we can approximate a continuous auction with a finite discretized version which is guaranteed to admit a mixed equilibrium by Nash's theorem. 
We refer the readers to the relevant discussion in \citet{FFGL13} and \citet{CP14}. 

Under the finite discretized version of the auction, a mixed Bayes Nash equilibrium is guaranteed to exist. 
It remains to show, however, that the space of bid profiles satisfying both sNUB and NOB is non-empty.

\begin{observation}
	Every S2PA with arbitrary monotone valuation functions admits a bid profile that satisfies both sNUB and NOB.
\end{observation}

\begin{proof}
	Let $S^*(\mathbf{v}) = (S_1^*(\mathbf{v}), \ldots, S_n^*(\mathbf{v}))$
	be an optimal allocation.
	Consider the bid profile $\mathbf{b}$, where
	$b_{ij} = \frac{v_i(S_i^*(v))}{|S_i^*(v)|}$ for $j \in S_i^*(v)$ and $0$ otherwise.
	Notice that each bidder $i$ wins the items she gets in $S^*(\mathbf{v})$, i.e., $S_i(\mathbf{b}) = S_i^*(v)$.
	Hence, $\sum_{j \in S_i(\mathbf{b})} b_{ij} = \sum_{j \in S^*_i(\mathbf{v})} b_{ij} = v_i(S_i^*(v)) = v_i(S_i(\mathbf{b}))$, showing that $\mathbf{b}$ satisfies NOB.
	Moreover, as $S^*_i(\mathbf{v})\backslash S_i(\mathbf{b}) = \emptyset$, $\mathbf{b}$ also satisfies sNUB.
\end{proof}


\section*{Acknowledgments}
We thank Noam Nisan for insightful comments regarding underbidding in the context of dominated strategies. We also thank Tomer Ezra for Example \ref{exmpl:xos_not_alpha_sm} regarding XOS valuations and $\alpha$-submodular valuations.


\bibliographystyle{plainnat}
\bibliography{BibFile}

\appendix
\def\appendixname{}





\section{Missing Proofs from Section \ref{sec:our_contribution}}

\subsection{S2PA with XOS Valuations and iNUB}
\label{app:s2pa_xos_inub}
\begin{theorem}
\label{thm:xos_1m_rev_guarant}
Every S2PA with XOS valuations and 
a bid profile $\mathbf{b}$ that satisfies iNUB is $(1,m)-$revenue guaranteed, where $m$ is the number of items.
\end{theorem}

\begin{proof}

We start with Inequality (\ref{eqn:sm_rev_guarant_4}):

\begin{eqnarray}
\label{eqn:xos_1m_rev_guarant_1}
\sum_{i=1}^{n} \sum_{j \in S_i(\mathbf{b})} p_j(\mathbf{b})
&\geq&
\sum_{i=1}^{n} 
\sum_{j \in S^*_i(\mathbf{v})} 
[v_i(S_i(\mathbf{b}) \cup \set{j}) - v_i(S_i(\mathbf{b})]
\end{eqnarray}
Consider the first term on the right hand side of Inequality (\ref{eqn:xos_1m_rev_guarant_1}). Let $v^*_i$ be the maximizing additive valuation of player $i$ with respect to her set $S^*_i(\mathbf{v})$. As $v_i$ is an XOS function, for every set $S' \subseteq [m]$ we have $v_i(S') \ge \sum_{j \in S'} v^*_i(\set{j})$. Hence, together with monoticity of $v_i$, we get:
\begin{eqnarray}
\nonumber
\sum_{i=1}^{n} 
\sum_{j \in S^*_i(\mathbf{v})} 
v_i(S_i(\mathbf{b}) \cup \set{j})
& \ge &
\sum_{i=1}^{n} 
\sum_{j \in S^*_i(\mathbf{v})} 
v_i(\set{j}) \\
\nonumber
& \ge &
\sum_{i=1}^{n} 
\sum_{j \in S^*_i(\mathbf{v})} 
v^*_i(\set{j}) \\
\nonumber
& = &
\sum_{i=1}^{n} v_i(S^*_i(\mathbf{v})) \\
\label{eqn:xos_1m_rev_guarant_2}
& = &
OPT(\mathbf{v})
\end{eqnarray}
Consider the second term on the right hand side of Inequality (\ref{eqn:xos_1m_rev_guarant_1}),
\begin{eqnarray}
\nonumber
\sum_{i=1}^{n} 
\sum_{j \in S^*_i(\mathbf{v})} 
v_i(S_i(\mathbf{b}))
& \le & 
\sum_{i=1}^{n} 
\sum_{j \in [m]}
v_i(S_i(\mathbf{b})) \\
\nonumber
& = &
\sum_{j \in [m]}
\sum_{i=1}^{n} 
v_i(S_i(\mathbf{b})) \\
\nonumber
& = & 
\sum_{j \in [m]} SW(\mathbf{b}, \mathbf{v}) \\
\label{eqn:xos_1m_rev_guarant_3}
& = & 
m \cdot SW(\mathbf{b}, \mathbf{v})
\end{eqnarray}
Combining Equations (\ref{eqn:xos_1m_rev_guarant_1}), (\ref{eqn:xos_1m_rev_guarant_2}) and (\ref{eqn:xos_1m_rev_guarant_3}), we get:
$$\sum_{i=1}^{n} \sum_{j \in S_i(\mathbf{b})} p_j(\mathbf{b})
\geq
OPT(\mathbf{v}) - m \cdot SW(\mathbf{b}, \mathbf{v})$$
\end{proof}

\begin{corollary}
In an S2PA with XOS valuations, 
for every 
joint distribution $\mathcal{F} \in \Delta(\mathcal{V}_1 \times \ldots \times \mathcal{V}_n)$, possibly correlated, 
every mixed Bayes Nash equilibrium that satisfies iNUB has expected social welfare at least $\frac{1}{m+1}$ of the expected optimal social welfare.
\end{corollary}
%
%
%

We next show an example of an instance with $4$ items and $2$ \xos\  bidders, where the PoA with NOB and iNUB is $1/2$, which is no better than the guarantee obtained with NOB alone.

\begin{example}
\label{ex:xos_nob_inub_1_2_poa}
There are two \xos\  bidders, $\set{1,2}$, and four items, $\set{x,y,z,w}$. Let,
\begin{eqnarray*}
v_1(S) &=& max \set{a^1(S), a^2(S)} ~~ \mbox{   and,} \\
v_2(S) &=& max \set{a^3(S), a^4(S)}
\end{eqnarray*}

where:
\begin{eqnarray*}
a^1 
&=& 
(a^1_x, a^1_y, a^1_z, a^1_w) 
=
(2,2,0,0), \\
a^2 
&=& 
(a^2_x, a^2_y, a^2_z, a^2_w) 
=
(0,0,1,1), \\
a^3 
&=& 
(a^3_x, a^3_y, a^3_z, a^3_w) 
=
(0,0,2,2), \\
a^4 
&=& 
(a^4_x, a^4_y, a^4_z, a^4_w) 
=
(1,1,0,0).
\end{eqnarray*}

The optimal allocation has welfare $8$, giving items $x$ and $y$ to player $1$ and items $z$ and $w$ to player $2$.
Consider the pure Nash equilibrium bid profile $\mathbf{b} = (b_1,b_2)$, where:
\begin{eqnarray*}
b_1 &=& (b_{1x},b_{1y},b_{1z},b_{1w}) = (0,0,1,1) ~~ \mbox{   and,} \\
b_2 &=& (b_{2x},b_{2y},b_{2z},b_{2w}) = (1,1,0,0)
\end{eqnarray*}

Player $1$ wins items $z$ and $w$ and player $2$ wins items $x$ and $y$.
$\mathbf{b}$ satisfies NOB and iNUB, and obtains welfare $4$, which equals half of $OPT$.
\end{example}

\subsection{S2PA with Monotone Valuations and iNUB}
\label{app:s2pa_mon_inub}
The following example shows that beyond subadditive valuations, the PoA can be arbitrarily bad under bid profiles satisfying iNUB.

\begin{example}
	\label{ex:mon_inub_bad_poa}
	$2$ items: $\set{x,y}$, $2$ single-minded bidders, who only derive value from the package of both items.
	Suppose $v_1(xy)=1$, $v_2(xy)=R$ (where $R$ is arbitrarily large), and the value for any strict subset of $xy$ is 0. 
	Consider the following bid profile (which is a PNE that adheres to iNUB): $b_{1x}=b_{1y}=R$, and $b_{2x}=b_{2y}=0$.
	Under this bid profile, both items go to agent 1, for a PoA $R$.

\end{example}

\section{Missing Proofs from Section \ref{subsec:val_class}}

\subsection{XOS Valuations are not \texorpdfstring{$\alpha-$}{}SM}
\label{app:xos_not_alpha_sm}
The following example shows that there exists an XOS function over identical items, sets $S \subset T$ and $j \notin T$ such that $v(j \mid S) < \alpha \cdot v(j \mid T)$ for every $\alpha > 0$.
\begin{example}
\label{exmpl:xos_not_alpha_sm}
Consider three identical items and valuation function $v$, where
$$
v(S)=
\begin{cases}
1, \mbox { if } |S| \in \{1,2\}\\
1.5, \mbox{ if } |S|=3
\end{cases}
$$
One can verify that $v$ is an XOS function.
Let $S$ be a set that contains a single item, and $T$ be a set that contains two items, such that $S \subset T$. 
For $j \notin T$ it holds that $v(j \mid S) = 0$, and $v(j \mid T) = 0.5$. Thus, $v(j \mid S) < \alpha \cdot v(j \mid T)$ for every $\alpha > 0$.
\end{example}

\subsection{\texorpdfstring{$\alpha-$}{}SM Valuations are not XOS}
\label{app:alpha_sm_not_xos}
The following example shows that for every $\alpha \in (0,1)$, there exists an $\alpha-$SM function over identical items
which is not XOS.

\begin{example}
Consider a setting with a set $M$ of three identical items, $0 < \alpha < 1$, and the following valuation function $v$:
$$
v(S)=
\begin{cases}
2,~~~~~~~~~ \mbox { if } |S|=1\\
2(1+\alpha), \mbox{ if } |S|=2\\
2(2+\alpha), \mbox{ if } |S|=3
\end{cases}
$$
One can verify that $v$ is $\alpha-$SM.
Let us assume by contradiction that $v$ is XOS, i.e., 
there is a set $\mathcal{L}$ of additive valuations $\{a_{\ell}(\cdot)\}_{\ell \in \mathcal{L}}$, such that for every set $S \subseteq [m]$, $v(S) = max_{\ell \in \mathcal{L}} a_{\ell}(S)$. 
For every $\ell$, it should hold that $a_{\ell}(S) \le 2(1+\alpha)$ for every $|S| = 2$. 
It follows that for every $\ell$, $a_{\ell}(M) \le 3(1+\alpha) < 2(2+\alpha) = v(M)$, in contradiction to $v$ being XOS.

\end{example}

\section{Missing Proofs from Section \ref{subsec:smooth}}
\label{app:extens_thm_proof}
Following is the proof of Theorem \ref{thm:smooth_poa_copm_i}. 
The proof is based on the proofs in \cite{R09} and \cite{RST17} with the necessary adjustments. 

\begin{proof} (of Theorem \ref{thm:smooth_poa_copm_i})
As the utility of each player is quasi-linear and payments are non negative, 
$$\sum_{i \in [n]}
\mathbb{E}_{\mathbf{b}}
\left[ 
u_i(\mathbf{b}, v_i)
\right] \le 
\mathbb{E}_{\mathbf{b}}
\left[ 
SW(\mathbf{b}, \mathbf{v})
\right]$$
By smoothness, 
for each $b$ in the support of $\mathbf{b}$, 
there exists a bid $b^*_i(\mathbf{v})  \in \mathcal{B}_i$ for each player $i$, s.t.:
$$\sum_{i \in [n]} u_i(b^*_i(\mathbf{v}), b_{-i},v_i) \ge \lambda OPT(\mathbf{v}) - \mu SW(b, \mathbf{v})$$
Since $\mathbf{b}$ is a CCE, we get from Definition \ref{def:ccne},
$$\mathbb{E}_{\mathbf{b}} 
\left[ 
u_i(\mathbf{b}, v_i)
\right]
\ge 
\mathbb{E}_{\mathbf{b}} 
\left[
u_i(b^*_i(\mathbf{v}),\mathbf{b}_{-i}, v_i)
\right]$$
Punting it all together and using linearity of expectation we get,
\begin{eqnarray}
\nonumber
\mathbb{E}_{\mathbf{b}}
\left[ 
SW(\mathbf{b}, \mathbf{v})
\right] 
& \ge &  
\sum_{i \in [n]}
\mathbb{E}_{\mathbf{b}}
\left[ 
u_i(\mathbf{b}, v_i)
\right] \\
\nonumber
& \ge & 
\sum_{i \in [n]}
\mathbb{E}_{\mathbf{b}} 
\left[
u_i(b^*_i(\mathbf{v}),\mathbf{b}_{-i}, v_i)
\right] \\
\nonumber
& \ge & 
\mathbb{E}_{\mathbf{b}}
\left[
\lambda OPT(\mathbf{v}) - \mu SW(\mathbf{b}, \mathbf{v})
\right] \\
\label{eqn:smooth_poa_copm_i_2}
& = &
\lambda OPT(\mathbf{v}) - \mu \mathbb{E}_{\mathbf{b}}
\left[ 
SW(\mathbf{b}, \mathbf{v})
\right]
\end{eqnarray}
Rearranging, we get: $\mathbb{E}_{\mathbf{b}} \left[ SW(\mathbf{b}, \mathbf{v}) \right] \ge \frac{\lambda}{1+\mu} OPT(\mathbf{v})$, as required.
\end{proof}

Following is the proof of Theorem \ref{thm:smooth_poa_incopm_i}.
The proof is based on the proofs in \cite{R12} and \cite{ST13} with the necessary adjustments.

\begin{proof} (of Theorem \ref{thm:smooth_poa_incopm_i})
We give a proof for pure Bayes Nash equilibrium. The proof for mixed Bayes Nash equilibrium 
follows by adding in a straightforward way 
another expectation over the random actions chosen in the strategy profile $\mathbf{\sigma}$.

Recall that the valuations $v_i$ are independent and that the strategy $\sigma_i(v_i)$ of player $i$ depends only on $v_i$. However, the smoothness definition (Definition \ref{def:smooth}) 
assumes that the hypothetical deviation of player $i$ depends on the full valuation profile. Since player $i$ doesn't know the valuation of the other players, she randomly samples an independent valuation profile $\mathbf{w} = (w_1, \ldots, w_n) \sim \mathcal{F}$ and decides on the hypothetical deviation $\mathbf{b}_i^*(v_i, \mathbf{w}_{-i})$ accordingly. 

As $\mathbf{\sigma}$ is a BNE of the auction, using Inequality (\ref{eqn:bne}) we get,
\begin{eqnarray*}
\mathbb{E}_{\mathbf{v}} 
\left[ 
u_i(\mathbf{\sigma}(\mathbf{v}), v_i)
\right]
& \ge & 
\mathbb{E}_{\mathbf{v}} 
\mathbb{E}_{\mathbf{w}} 
\left[
u_i(\mathbf{b}_i^*(v_i, \mathbf{w}_{-i}),\mathbf{\sigma}_{-i}(\mathbf{v}_{-i}), v_i)
\right] \\
& = &
\mathbb{E}_{\mathbf{v}} 
\mathbb{E}_{\mathbf{w}} 
\left[
u_i(\mathbf{b}_i^*(w_i, \mathbf{w}_{-i}),\mathbf{\sigma}_{-i}(\mathbf{v}_{-i}), w_i)
\right] \\
& = &
\mathbb{E}_{\mathbf{v}} 
\mathbb{E}_{\mathbf{w}} 
\left[
u_i(\mathbf{b}_i^*(\mathbf{w}),\mathbf{\sigma}_{-i}(\mathbf{v}_{-i}), w_i)
\right] 
\end{eqnarray*}

The equality follows by renaming due to independence.
Now, let us sum over all players and then use the smoothness Inequality (\ref{eqn:smooth}) and linearity of expectation to get,
\begin{eqnarray}
\nonumber
\mathbb{E}_{\mathbf{v}} 
\left[ 
\sum_{i \in [n]}
u_i(\mathbf{\sigma}(\mathbf{v}), v_i)
\right]
& \ge & 
\mathbb{E}_{\mathbf{v}} 
\mathbb{E}_{\mathbf{w}} 
\left[
\sum_{i \in [n]}
u_i(\mathbf{b}_i^*(\mathbf{w}),\mathbf{\sigma}_{-i}(\mathbf{v}_{-i}), w_i)
\right] \\
\nonumber
& \ge & 
\mathbb{E}_{\mathbf{v}} 
\mathbb{E}_{\mathbf{w}} 
\left[
\lambda OPT(\mathbf{w}) - \mu SW(\mathbf{\sigma}(\mathbf{v}), \mathbf{v})
\right] \\
\label{eqn:smooth_poa_incopm_i_2}
& = &
\lambda \mathbb{E}_{\mathbf{w}} 
\left[
OPT(\mathbf{w})
\right] - 
\mu \mathbb{E}_{\mathbf{v}} 
\left[ SW(\mathbf{\sigma}(\mathbf{v}), \mathbf{v})
\right]
\end{eqnarray}

As the utility of each player is quasi-linear and payments are non negative, 
$$\mathbb{E}_{\mathbf{v}} 
\left[ 
\sum_{i \in [n]}
u_i(\mathbf{\sigma}(\mathbf{v}), v_i)
\right]
 \le 
\mathbb{E}_{\mathbf{v}}
\left[ 
SW(\mathbf{\sigma}(\mathbf{v}), \mathbf{v})
\right]$$
Hence, 
$$\mathbb{E}_{\mathbf{v}}
\left[ 
SW(\mathbf{\sigma}(\mathbf{v}), \mathbf{v})
\right] \ge 
\lambda \mathbb{E}_{\mathbf{v}} 
\left[
OPT(\mathbf{v})
\right] - 
\mu \mathbb{E}_{\mathbf{v}} 
\left[ SW(\mathbf{\sigma}(\mathbf{v}), \mathbf{v})
\right]$$
The proof follows by rearranging.
\end{proof}

\section{Missing Proofs from Section \ref{sec:s2pa_nub}}
\label{app:ex_pne_snub_not_inub}
Following is an example of an S2PA with submodular bidders that admits a PNE satisfying sNUB but not iNUB.

\begin{example}
\label{exmpl:ex_pne_snub_not_inub}
Two submodular bidders: $\set{1,2}$, and three items: $\set{x,y,z}$.
$$v_1(x) = 5, 
v_1(y) = 5, 
v_1(z) = 10, 
v_1(xy) = 10, 
v_1(xz) = 15, 
v_1(yz) = 15, 
v_1(xyz) = 16$$
$$v_2(x) = 8, 
v_2(y) = 8, 
v_2(z) = 15, 
v_2(xy) = 14, 
v_2(xz) = 15, 
v_2(yz) = 15, 
v_2(xyz) = 15$$

The optimal allocation is: 
$S_1^* = \{xy\}, ~
S_2^* = \{z\}$, 
$OPT = 10+15 = 25$.

One can verify that the following bid profile $\mathbf{b}$ is a PNE: 
$$b_{1x} = 3, ~
b_{1y} = 3, ~
b_{1z} = 8$$
$$b_{2x} = 8, ~
b_{2y} = 8, ~
b_{2z} = 2,$$ 
and the obtained allocation under $\mathbf{b}$ is:
$$S_1(b) = \{z\}, ~
S_2(b) = \{xy\}, ~
SW = 10+14 = 24.$$


We first show that $\bids$ satisfies sNUB. Indeed, 
$$6 = b_{1x}+b_{1y} \ge v_1(xy \mid z) = 6$$
and also,
$$2 = b_{2z} \ge v_2(z \mid xy) = 1.$$
However, $\bids$ does not satisfy iNUB, since
$$3 = b_{1x} < v_1(x \mid z) = 5.$$

\end{example}

%
%
%
%
%

\section{Revenue Guaranteed Auctions in Settings with Incomplete Information}
\label{app:rev_guarant_incomp}
In the paper we defined both revenue guaranteed auctions and no-underbidding bid profile for full information with respect to a bid space $\mathcal{B}'$, which is the space of all bid profiles satisfying sNUB. We then used extension theorem to prove positive results for incomplete information on strategy space $\sigma : \mathcal{V}_1 \times \ldots \times \mathcal{V}_n \rightarrow \Delta(\mathcal{B}')$, i.e. for all strategies supported by bids $\sigma(\mathbf{v})$ satisfying sNUB. However, there might be strategies 
which satisfy no-underbidding only in expectation and yet guarantee lower bound on the expected revenue. Moreover, there might be auctions that are revenue guaranteed only in expectation. To prove results in incomplete information for such broader cases, we give here incomplete information definition to both revenue guaranteed auctions and no-underbidding strategy profiles.

\begin{definition} [\textbf{Revenue guaranteed auction of incomplete information}]
\label{def:rev_guarant_incomp}
An incomplete information auction is $(\gamma,\delta)-$revenue guaranteed for some $0 \le \gamma \le \delta \le 1$ with respect to a strategy space $\Sigma' \subseteq \Delta(\Sigma_1 \times \ldots \times \Sigma_n)$, if for every joint distribution $\mathcal{F} \in \Delta(\mathcal{V}_1 \times \ldots \times \mathcal{V}_n)$, possibly correlated,
and for any strategy profile $\mathbf{\sigma} \in \Sigma'$,
the expected revenue of the auction is at least 
$\gamma ~
\mathbb{E}_{\mathbf{v} \sim \mathcal{F}}
\left[
OPT(\mathbf{v})
\right]
- \delta ~
\mathbb{E}_{\mathbf{v} \sim \mathcal{F}} ~
\mathbb{E}_{\mathbf{b} \sim \sigma(\mathbf{v})} 
\left[
SW(\mathbf{b}, \mathbf{v})
\right]$.
\end{definition}

\begin{theorem}
\label{thm:rev_guarant_poa_incopm_extended}
If an incomplete information auction is $(\gamma,\delta)-$revenue guaranteed with respect to a strategy space $\Sigma' \subseteq \Delta(\Sigma_1 \times \ldots \times \Sigma_n)$,
then for every 
joint distribution $\mathcal{F} \in \Delta(\mathcal{V}_1 \times \ldots \times \mathcal{V}_n)$, possibly correlated, 
and every strategy profile 
$\mathbf{\sigma} \in \Sigma'$,
in which the expected sum of players utility is non-negative,  
the expected social welfare is at least $\frac{\gamma}{1+\delta}$ of the expected optimal social welfare. 
\end{theorem}

\begin{proof}
As the utility of each player is quasi-linear and the expected sum of player utilities is non-negative, we use linearity of expectation and get,
\begin{eqnarray*}
0 
& \le &  
\sum_{i \in [n]}
\mathbb{E}_{\mathbf{v} \sim \mathcal{F}} ~
\mathbb{E}_{\mathbf{b} \sim \sigma(\mathbf{v})} 
\left[ 
u_i(\mathbf{b}, v_i)
\right] \\
& = & 
\sum_{i \in [n]}
\mathbb{E}_{\mathbf{v} \sim \mathcal{F}} ~
\mathbb{E}_{\mathbf{b} \sim \sigma(\mathbf{v})}
\left[ 
v_i(S_i(\mathbf{b}))
\right] - 
\sum_{i \in [n]}
\mathbb{E}_{\mathbf{v} \sim \mathcal{F}} ~
\mathbb{E}_{\mathbf{b} \sim \sigma(\mathbf{v})}
\left[ 
P_i(\mathbf{b})
\right] \\
& = & 
\mathbb{E}_{\mathbf{v} \sim \mathcal{F}} ~
\mathbb{E}_{\mathbf{b} \sim \sigma(\mathbf{v})}
\left[ 
SW(\mathbf{b}, \mathbf{v})
\right] - 
\sum_{i \in [n]}
\mathbb{E}_{\mathbf{v} \sim \mathcal{F}} ~
\mathbb{E}_{\mathbf{b} \sim \sigma(\mathbf{v})}
\left[ 
P_i(\mathbf{b})
\right]
\end{eqnarray*}
By the $(\gamma,\delta)-$revenue guaranteed property for incomplete information, 
$$\sum_{i \in [n]}
\mathbb{E}_{\mathbf{v} \sim \mathcal{F}} ~
\mathbb{E}_{\mathbf{b} \sim \sigma(\mathbf{v})}
\left[ 
P_i(\mathbf{b})
\right] 
\ge
\gamma ~
\mathbb{E}_{\mathbf{v} \sim \mathcal{F}}
\left[
OPT(\mathbf{v})
\right]
- \delta ~
\mathbb{E}_{\mathbf{v} \sim \mathcal{F}} ~
\mathbb{E}_{\mathbf{b} \sim \sigma(\mathbf{v})} 
\left[
SW(\mathbf{b}, \mathbf{v})
\right]
$$

Punting it all together, we get,
\begin{eqnarray}
\nonumber
0 
& \le & 
\mathbb{E}_{\mathbf{v} \sim \mathcal{F}} ~
\mathbb{E}_{\mathbf{b} \sim \sigma(\mathbf{v})}
\left[ 
SW(\mathbf{b}, \mathbf{v})
\right] - 
\sum_{i \in [n]}
\mathbb{E}_{\mathbf{v} \sim \mathcal{F}} ~
\mathbb{E}_{\mathbf{b} \sim \sigma(\mathbf{v})}
\left[ 
P_i(\mathbf{b})
\right] \\
\nonumber
& \le &
\mathbb{E}_{\mathbf{v} \sim \mathcal{F}} ~
\mathbb{E}_{\mathbf{b} \sim \sigma(\mathbf{v})}
\left[ 
SW(\mathbf{b}, \mathbf{v})
\right] - 
\gamma ~
\mathbb{E}_{\mathbf{v} \sim \mathcal{F}}
\left[
OPT(\mathbf{v})
\right]
+ \delta ~
\mathbb{E}_{\mathbf{v} \sim \mathcal{F}} ~
\mathbb{E}_{\mathbf{b} \sim \sigma(\mathbf{v})} 
\left[
SW(\mathbf{b}, \mathbf{v})
\right] \\
& = &
\nonumber
(1+\delta) ~ \mathbb{E}_{\mathbf{v} \sim \mathcal{F}} ~
\mathbb{E}_{\mathbf{b} \sim \sigma(\mathbf{v})} 
\left[
SW(\mathbf{b}, \mathbf{v})
\right] - 
\gamma ~
\mathbb{E}_{\mathbf{v} \sim \mathcal{F}}
\left[
OPT(\mathbf{v})
\right]
\end{eqnarray}

Rearranging, we get: 
$\mathbb{E}_{\mathbf{v} \sim \mathcal{F}} ~
\mathbb{E}_{\mathbf{b} \sim \sigma(\mathbf{v})} 
\left[
SW(\mathbf{b}, \mathbf{v})
\right]
\ge 
\frac{\gamma}{1+\delta} ~
\mathbb{E}_{\mathbf{v} \sim \mathcal{F}}
\left[
OPT(\mathbf{v})
\right]$, as required.
\end{proof}

\begin{definition} [\textbf{Set No underbidding in expectation}]
\label{def:sNUB_incomp}
A strategy profile $\mathbf{\sigma} \in \Delta(\Sigma_1 \times \ldots \times \Sigma_n)$ satisfies sNUB in expectation if 
for every 
joint distribution $\mathcal{F} \in \Delta(\mathcal{V}_1 \times \ldots \times \mathcal{V}_n)$, possibly correlated,
and for every player $i$ 
the following holds,
\begin{eqnarray*}
\mathbb{E}_{\mathbf{v} \sim \mathcal{F} \mid {v_i}} ~
\mathbb{E}_{\mathbf{b} \sim \sigma(v)}
\left[ 
\sum_{j \in S^*_i(\mathbf{v})\backslash S_i(\mathbf{b})} b_{ij}
\right]
& \ge &
\mathbb{E}_{\mathbf{v} \sim \mathcal{F} \mid {v_i}} ~
\mathbb{E}_{\mathbf{b} \sim \sigma(v)}
\left[ 
v_i \left(~ S^*_i(\mathbf{v})\backslash S_i(\mathbf{b}) \mid S_i(\mathbf{b})~ \right)
\right]
\end{eqnarray*}
\end{definition}

\begin{theorem}
\label{thm:mon_11_rev_guarant_incomp}
An incomplete information S2PA with monotone valuation functions is $(1,1)-$revenue guaranteed with respect to 
strategy profiles satisfying sNUB in expectation.
\end{theorem}

\begin{proof}
We start with Lemma \ref{lem:s2pa_rev_bids}, use linearity of expectation and
the fact that $\mathbf{\sigma}$ satisfies sNUB in expectation,

\begin{eqnarray*}
\mathbb{E}_{\mathbf{v} \sim \mathcal{F}} ~
\mathbb{E}_{\mathbf{b} \sim \sigma(\mathbf{v})} 
\left[
\sum_{i=1}^{n} \sum_{j \in S_i(\mathbf{b})} p_j(\mathbf{b})
\right]
& \ge &
\mathbb{E}_{\mathbf{v} \sim \mathcal{F}} ~
\mathbb{E}_{\mathbf{b} \sim \sigma(\mathbf{v})} 
\left[
\sum_{i=1}^{n} \sum_{j \in S^*_i(\mathbf{v}) \backslash S_i(\mathbf{b})} b_{ij} 
\right] \\
& = &
\sum_{i=1}^{n}
\mathbb{E}_{\mathbf{v} \sim \mathcal{F}}
\mathbb{E}_{\mathbf{b} \sim \sigma(\mathbf{v})} 
\left[
\sum_{j \in S^*_i(\mathbf{v}) \backslash S_i(\mathbf{b})} b_{ij} 
\right] \\
& \ge &
\sum_{i=1}^{n}
\mathbb{E}_{\mathbf{v} \sim \mathcal{F}}
\mathbb{E}_{\mathbf{b} \sim \sigma(v)}
\left[ 
v_i \left(~ S^*_i(\mathbf{v})\backslash S_i(\mathbf{b}) \mid S_i(\mathbf{b})~ \right)
\right] \\
& = &
\sum_{i=1}^{n} 
\mathbb{E}_{\mathbf{v} \sim \mathcal{F}} ~
\mathbb{E}_{\mathbf{b} \sim \sigma(\mathbf{v})} 
\left[ 
v_i(S_i(\mathbf{b}) \cup S^*_i(\mathbf{v})) - v_i(S_i(\mathbf{b})
\right] \\
&\ge &
\sum_{i=1}^{n} 
\mathbb{E}_{\mathbf{v} \sim \mathcal{F}} ~
\mathbb{E}_{\mathbf{b} \sim \sigma(\mathbf{v})} 
\left[ 
v_i(S^*_i(\mathbf{v})) - v_i(S_i(\mathbf{b})
\right] \\
&=&
\mathbb{E}_{\mathbf{v} \sim \mathcal{F}} ~
\mathbb{E}_{\mathbf{b} \sim \sigma(\mathbf{v})} 
\left[ 
\sum_{i=1}^{n} 
\left[
v_i(S^*_i(\mathbf{v})) - v_i(S_i(\mathbf{b})
\right]
\right] \\
&=&
\mathbb{E}_{\mathbf{v} \sim \mathcal{F}}
\left[ OPT(\mathbf{v}) \right] - 
\mathbb{E}_{\mathbf{v} \sim \mathcal{F}} ~
\mathbb{E}_{\mathbf{b} \sim \sigma(\mathbf{v})} 
\left[ 
SW(\mathbf{b},\mathbf{v})
\right]
\end{eqnarray*}
The last Inequality follows from monoticity of valuations.
\end{proof}

\begin{corollary}
\label{corr:mon_snub_bpoa_1_2_incomp}
In an incomplete information S2PA with monotone valuations, 
for every 
joint distribution $\mathcal{F} \in \Delta(\mathcal{V}_1 \times \ldots \times \mathcal{V}_n)$, possibly correlated, 
every mixed Bayes Nash equilibrium that satisfies sNUB in expectation has expected social welfare at least $\frac{1}{2}$ of the expected optimal social welfare.
\end{corollary}


\end{document}